\documentclass[a4paper,11pt]{paper}

\usepackage{fullpage,latexsym,amsthm,amsmath,amssymb,url}
\usepackage[ruled, linesnumbered]{algorithm2e}
\usepackage{tikz}
\usepackage{wasysym,marvosym}
\usepackage{comment}

\usepackage{enumerate}
\usepackage[shortlabels]{enumitem}
\usepackage{hyperref}
\usetikzlibrary{arrows,decorations.pathreplacing}

\usepackage{boxedminipage}
\usepackage{marvosym}
\usepackage{stmaryrd}
\usepackage[all]{xy}
\usepackage[normalem]{ulem}
\usepackage{soul}

\input colordvi

\newtheorem{lemma}{Lemma}
\newtheorem{theorem}{Theorem}
\newtheorem{proposition}[theorem]{Proposition}

\newtheorem{corollary}{Corollary}
\newtheorem{definition}{Definition}
\newtheorem{claim}{Claim}

\newcommand{\cw}{\mathsf{ctw}}

\newcommand{\ola}{\textsc{Ola}}
\newcommand{\vars}{\mathrm{vars}}
\newcommand{\cls}{\mathrm{cls}}
\newcommand{\fstsmt}[2]{#1_{\leq #2}}

\newcommand{\lstsmt}[2]{#1_{> #2}}

\newcommand{\onesmt}[2]{#1_{#2}}
\newcommand{\cutvector}[2]{\mathsf{cuts}\langle #1,#2\rangle}
\newcommand{\Oh}{\mathcal{O}}

\def\cqedsymbol{\ifmmode$\lrcorner$\else{\unskip\nobreak\hfil
\penalty50\hskip1em\null\nobreak\hfil$\lrcorner$
\parfillskip=0pt\finalhyphendemerits=0\endgraf}\fi} 
\newcommand{\cqed}{\renewcommand{\qed}{\cqedsymbol}}

\newcommand{\blue}[1]{#1}
\newcommand{\red}[1]{#1}

\newcounter{col} \newcounter{row}
\newlength{\shift} \newlength{\tisize}

\title{Exploring the complexity of layout parameters in tournaments and semi-complete digraphs\thanks{The research of F. Barbero and C. Paul is supported by the DE-MO-GRAPH project ANR-16-CE40-0028. The research of Mi. Pilipczuk is supported by Polish National Science Centre grant UMO-2013/11/D/ST6/03073. Mi. Pilipczuk is also supported by the Foundation for Polish Science via the START stipend programme.}}
\author{Florian Barbero \thanks{LIRMM, Universit\'e de Montpellier, France, \texttt{florian.barbero@lirmm.fr}} 
   \and Christophe Paul \thanks{LIRMM, CNRS, Universit\'e de Montpellier, France, \texttt{christophe.paul@lirmm.fr}} 
   \and Micha\l{}' Pilipczuk \thanks{University of Warsaw, Poland, \texttt{michal.pilipczuk@mimuw.edu.pl}} }

\date{\today}

\begin{document}

\maketitle

\begin{abstract}
A simple digraph is \emph{semi-complete} if for any two of its vertices $u$ and $v$, 
at least one of the arcs $(u,v)$ and $(v,u)$ is present.
We study the complexity of computing two layout parameters of semi-complete digraphs: cutwidth and optimal linear arrangement (\ola{}).
We prove that:
\begin{itemize}[nosep]
\item Both parameters are $\mathsf{NP}$-hard to compute
and the known exact and parameterized algorithms for them have essentially optimal running times, assuming the Exponential Time Hypothesis.
\item The cutwidth parameter admits a quadratic Turing kernel, whereas it does not admit any polynomial kernel unless $\mathsf{NP}\subseteq \mathsf{coNP}/\textrm{poly}$.
By contrast, \ola{} admits a linear kernel.
\end{itemize}
These results essentially complete the complexity analysis of computing cutwidth and \ola{} on semi-complete digraphs.
Our techniques can be also used to analyze the sizes of minimal obstructions for having small cutwidth under the induced subdigraph relation.
\end{abstract}

\tableofcontents

\newpage
\section{Introduction}\label{sec:intro}

A directed graph (digraph) is \emph{simple} if it does not contain a self-loop or multiple arcs with the same head and tail.
A simple digraph is \emph{semi-complete} if for any pair of its vertices $u$ and $v$, at least one of the arcs $(u,v)$ or $(v,u)$ is present.
If moreover exactly one of them is present for each pair $u,v$, then a semi-complete digraph is called a \emph{tournament}.
Tournaments and semi-complete digraphs \blue{form rich and interesting subclasses} of directed graphs; we refer to the book of Bang-Jensen and Gutin~\cite{Bang-Jensen-Gutin} for an overview.

We study two layout parameters for tournaments and semi-complete digraphs: cutwidth and optimal linear arrangement (\ola{}).
Suppose $\pi$ is an ordering of the vertices of a digraph $D$.
With each prefix of $\pi$ we associate a \emph{cut} defined as the set of arcs with head in the prefix and tail outside of it.
The \emph{width} of $\pi$ is defined as the maximum size among the cuts associated with the prefixes of $\pi$.
The \emph{cutwidth} of $D$, denoted $\cw(D)$, is the minimum width among orderings of the vertex set of $D$.
Optimal linear arrangement (\ola{}) is defined similarly, but when defining the width of $\pi$, called in this context the \emph{cost} of $\pi$, 
we take the sum of the cutsizes associated with prefixes, instead of the maximum.
Then the \ola{}-cost of a digraph $D$, denoted $\ola(D)$, is the minimum cost among vertex orderings of $D$.

\paragraph{Known results.}
The study of cutwidth in the context of tournaments and semi-complete digraphs started with the work of Chudnovsky, Fradkin, and Seymour~\cite{CFS12,CS11,FradkinS15}, 
who identified this layout parameter as the \blue{accurate}
dual notion to immersions in semi-complete digraphs.
In particular, it is known that excluding a fixed digraph as an immersion yields a constant upper bound on the cutwidth of a semi-complete digraph~\cite{CFS12,Pilipczuk13}.
Due to this connection, cutwidth played a pivotal role in the proof of Chudnovsky and Seymour that the immersion order is a well quasi-order on tournaments~\cite{CS11}. 

The algorithmic properties of cutwidth were preliminarily investigated by Chudnovsky, Fradkin, and Seymour~\cite{CFS12,FradkinS15,Fra11}.
In Fradkin's PhD thesis~\cite{Fra11}, several results on the tractability of computing the cutwidth are presented.
In particular, it is shown that the cutwidth of a tournament can be computed optimally by just sorting vertices according to their outdegrees, whereas in semi-complete digraphs
a similar approach yields a polynomial-time $2$-approximation algorithm. The problem becomes NP-hard on super-tournaments, that is, when multiple parallel arcs are allowed.
Later, the third author together with Fomin proposed a parameterized algorithm for computing the cutwidth of a semi-complete digraph with running time $2^{\Oh(\sqrt{k\log k})}\cdot n^2$~\cite{FominP13,Pil13}, 
where $n$ is the number of vertices and $k$ is the target width.
Using the same techniques, \ola{} in semi-complete digraphs can be solved in time $2^{\Oh(k^{1/3}\sqrt{\log k})}\cdot n^2$~\cite{FominP13,Pil13}, where $k$ is the target cost.
It was left open whether the running times of these parameterized algorithms are optimal~\cite{Pil13}.
In fact, even settling the $\mathsf{NP}$-hardness of computing cutwidth and \ola{} in semi-complete digraphs was open~\cite{Fra11,Pil13}.

\paragraph{Our contribution.} 
We study two aspects of the computational complexity of computing cutwidth and \ola{} of semi-complete digraphs: optimality of parameterized algorithms and kernelization. First, we prove that these problems are $\mathsf{NP}$-hard and we provide almost tight lower bounds for the running times of algorithms solving them, based on the Exponential Time Hypothesis (ETH). 
\blue{More precisely, assuming ETH we prove that the cutwidth of a semi-complete digraph cannot be computed in time $2^{o(n)}$ nor in time $2^{o(\sqrt{k})}\cdot n^{\Oh(1)}$. The same arguments yields $2^{o(n)}$ and $2^{o(k^{1/3})}\cdot n^{\Oh(1)}$ lower bounds on the time to compute the \ola{}-cost of a semi-complete digraph. (See Theorem~\ref{thm:lower-bound} in Section~\ref{sec:lower-bound} for a precise statement.)} It follows that the known parameterized algorithms of Fomin and Pilipczuk~\cite{FominP13} are optimal under ETH, up to $\sqrt{\log k}$ factor in the exponent. Note that both cutwidth and \ola{}-cost can be computed in time $2^n\cdot n^{\Oh(1)}$ using standard dynamic programming on subsets, so we obtain tight lower bounds also for exact exponential-time algorithms. Interestingly, the lower bounds for both problems are based on the same reduction. 

\blue{
Next we turn our attention to the kernelization complexity of  computing the two layout parameters (see Section~\ref{sec:turing-kernel}).
We first observe that, when respectively parameterized by the target cost and the target width, the two problems behave differently. 
As far as \ola{} is concerned, we prove that there is a polynomial-time algorithm that, given an arbitrary digraph $D$ and a positive integer $c$, 
either correctly concludes that $\ola(D)>c$, or finds a digraph $D'$ on at most $2c$ vertices such that $\ola(D')=\ola(D)$. 
That is, computing the \ola{}-cost of an arbitrary digraphs admits a kernel of size linear in the target cost (see Theorem~\ref{thm:ola-kernel}). 
On the other hand, a simple \textsf{AND}-composition~\cite{Dru15,CFK14} shows that, under the assumption that $\mathsf{NP} \not \subseteq \mathsf{coNP}/\textrm{poly}$, 
such a polynomial size kernel is unlikely to exist for the problem of computing the cutwidth of a semi-complete digraph (see Theorem~\ref{thm:no-kernel}).
}

\blue{
Surprisingly, pre-processing for computing the cutwidth of a semi-complete digraph turns out to be efficient if we consider the alternative notion of kernelization called {\em{Turing kernelization}}. In this framework, which has also been studied intensively in the literature (cf. the discussion in~\cite{CFK14}), it is not required that the instance at hand is reduced to one equivalent small instance,  but rather that the whole problem can be solved in polynomial time assuming oracle access to an algorithm solving instances of size bounded by a function of the parameter. }

\blue{
We design a polynomial-time algorithm that given a semi-complete digraph $D$ and integer $c$, either correctly concludes that $\cw(D)>c$ or  outputs a list of at most $n$ induced subdigraphs $D_1,\ldots,D_{\ell}$ of $D$, each with at most $\Oh(c^2)$ vertices, such that $\cw(D) \leq c$ if and only if $\cw(D_i)\leq c$ for each $i\in \{1,2,\ldots,\ell\}$ (see Theorem~\ref{thm:turing-kernel}). Observe that this algorithm is a so-called {\em{{\sc{and}}-Turing kernel}}, meaning that it just computes the output list without any oracle calls,  and the answer to the input instance is the conjunction of the answers to the output small instances. This places the problem of computing the cutwidth of a semi-complete digraph among very few known examples of natural problems where classic and Turing kernelization have different computational power~\cite{BFF12,Jan17,GW16,SKM12,TTV17}. Moreover, this is the first \blue{known to us} polynomial {\sc{and}}-Turing kernel for a natural problem:  examples of Turing kernelization known in the literature are either {\sc{or}}-Turing kernels~\cite{BFF12,GW16,SKM12},  or adaptative kernels that fully exploit the oracle model~\cite{Jan17,TTV17}. As separating classic and Turing kernelization is arguably one of the most important complexity-theoretical open problems within parameterized complexity~\cite{DowneyF13,HKS13,CFK14}, we find this new example intriguing. 
}

The proof of the Turing kernel relies on the notion of a {\em{lean ordering}}; see e.g.~\cite{GPR16,Tho90}.
Intuitively, a vertex ordering is lean if it is tight with respect to cut-flow duality: there are systems of arc-disjoint paths which certify that cutsizes along the ordering cannot be improved.
Lean orderings and decompositions are commonly used in the analysis of obstructions for various width notions, as well as for proving well quasi-order results.
In particular, the concept of a lean ordering for cutwidth of digraphs was used by Chudnovsky and Seymour in their proof that the immersion order is a well quasi-order on tournaments~\cite{CS11}.

As a byproduct of our approach to proving Theorem~\ref{thm:turing-kernel}, we obtain also polynomial upper bounds on the sizes of minimal obstructions to having small cutwidth. For a positive integer $c$, a digraph $D$ is called \emph{$c$-cutwidth-minimal} if the cutwidth of $D$ is at least $c$,  but the cutwidth of every proper induced subdigraph of $D$ is smaller than $c$. We show that every $c$-cutwidth-minimal semi-complete digraph has at most $\Oh(c^2)$ vertices (see Theorem~\ref{thm:obstructions-semi}) and that every $c$-cutwidth-minimal tournament has at most $2c+2\lceil \sqrt{2c}\rceil+1$ vertices~\ref{thm:obstructions-tour}. See Section~\ref{sec:obstructions}. We remark that from the well-quasi order result of Chudnovsky and Seymour~\cite{CS11}, it follows the number of minimal {\em{immersion}} obstructions for tournaments of cutwidth at most $c$ is finite. However, this holds only for tournaments, yields a non-explicit upper bound on obstruction sizes, and applies to immersion and not induced subdigraph obstructions. As discussed in Section~\ref{sec:obstructions}, these bounds  have direct algorithmic applications for parameterized graph modification problems related to cutwidth, e.g., $c$-{\sc Cutwidth Vertex Deletion}: remove at most $k$ vertices from a given digraph to obtain a digraph of cutwidth at most $c$.

\paragraph{Organization.}
In Section~\ref{sec:prelims} we establish basic notation and recall some background from complexity theory, concerning kernelization and the Exponential Time Hypothesis.
In Section~\ref{sec:polynomial} we give polynomial-time algorithms for computing the cutwidth and the \ola-cost of a tournament, as well as $2$-approximation algorithms for these parameters on semi-complete digraphs.
Section~\ref{sec:turing-kernel} is devoted to the kernelization complexity. We first describe the complexity of computing the cutwidth and the \ola-cost of a semi-complete digraph under the classic notion of
kernelization, and then we give a quadratic Turing kernel for the cutwidth case.
In Section~\ref{sec:obstructions} we utilize the tools developed in the previous sections to give upper bounds on the sizes of $c$-cutwidth-minimal tournaments and semi-complete digraphs.
Section~\ref{sec:lower-bound} focuses on lower bounds under ETH: we prove that in the semi-complete setting, both computing the cutwidth and the \ola-cost are $\mathsf{NP}$-complete problems, and the
running times of known algorithms for them are essentially tight under ETH.
The shape of the lower bound construction suggests that the problem may be tractable if almost all vertices are not incident to any symmetric arcs (two oppositely oriented arcs between the same pair of vertices).
Indeed, in Section~\ref{sec:pure} we present an appropriate parameterization that captures this scenario, and prove a fixed-parameter tractability result for it.
Section~\ref{sec:conclusions} concludes the paper by giving some final remarks.


\section{Preliminaries and basic results}\label{sec:prelims}

\paragraph{Notations.}  
We use standard graph notation for digraphs. For a digraph $D$, the vertex and edge sets of $D$ are denoted by $V(D)$ and $E(D)$, respectively. 
For $X,Y \subseteq V(D)$, we denote $E(X,Y)= \{(u,v) \in E(D)\colon u \in X, v \in Y\}$.
The {\em{subdigraph induced}} by $X$ is the digraph $D[X]=(X,E(X,X))$.
Note that an induced subdigraph of a tournament is also a tournament, likewise for semi-complete digraphs. 
The \emph{inneighborhood} of a vertex $u$ is $N_D^-(u) = \{v \in V(D)\colon (v,u) \in E(D)\}$, a vertex $v \in N_D^-(u)$ being an \emph{inneighbour} of $u$. 
The \emph{indegree} of $u$ is $d_D^-(u) = |N_D^-(u)|$. We may drop the subscript when it is clear from the context. 
We define similarly the \emph{outneighborhood} $N^+(u)$ and the \emph{outdegree} $d^+(u)$.
All digraphs considered in this paper are simple, i.e., they do not contain a self-loop or multiple arcs with the same head and tail.
For definitions of tournaments and semi-complete digraphs, see the first paragraph of \blue{Section~\ref{sec:intro}.}
If present in a digraph, the arcs $(u,v)$ and $(v,u)$ are called \emph{symmetric arcs}. 

For two integers $p \leq p'$, let $[p,p'] \subseteq \mathbb{Z}$ be the set of integers between $p$ and $p'$. 
If $p < p'$, we set $[p',p] = \emptyset$ by convention. A {\em{vertex ordering}} of a digraph $D$ is a bijective mapping $\pi\colon V(D)\rightarrow [1,n]$, where $n=|V(D)|$. 
A vertex $u \in V(D)$ is \emph{at position $i$ in $\pi$} if $\pi(u) = i$. 
We denote this unique vertex by $\onesmt{\pi}{i}$. 
The {\em{prefix}} of length $i$ of $\pi$ is $\fstsmt{\pi}{i}= \{\onesmt{\pi}{j} \colon j \in [1,i]\}$; 
we set $\fstsmt{\pi}{i} = \emptyset$ when $i \leq 0$, and $\fstsmt{\pi}{i} = V(D)$ when $n \leq i$. 
We extend this notation to prefixes and suffixes of orderings naturally, e.g., $\lstsmt{\pi}{i} = V(D)\setminus \fstsmt{\pi}{i}$ is the set of the last $n-i$ vertices in $\pi$.
\blue{The notions of restriction of an ordering to a subset of vertices and of concatenation of orderings are defined naturally.}

An arc $(\onesmt{\pi}{i},\onesmt{\pi}{j}) \in E(D)$ is a \emph{feedback arc for $\pi$} if $i > j$, that is, if $\onesmt{\pi}{i}$ is after $\onesmt{\pi}{j}$ in $\pi$. 
Given a digraph $D=(V,E)$, an ordering $\pi$ of $V$ and an integer $i$, we define the \emph{cut $E^i_\pi$} as the set of feedback arcs $E(\lstsmt{\pi}{i},\fstsmt{\pi}{i})$. 
The tuple $\cutvector{D}{\pi} = (|E^0_\pi|, |E^1_\pi|, \ldots, |E^n_\pi|)$ is called the \emph{cut vector} of $\pi$, and we denote $\cutvector{D}{\pi}(i) = |E_\pi^i|$.
Let $\preceq$ be the product order on tuples: for $n$-tuples $A,B$, we have $A \preceq B$ iff $A(i) \leq B(i)$ for all $i \in [0,n]$. 
We define $A \prec B$ as $A \preceq B$ and $A \neq B$.
We say that a vertex ordering $\pi$ is \emph{minimum for $D$} if for all vertex orderings $\pi'$ of $D$ we have $\cutvector{D}{\pi} \preceq \cutvector{D}{\pi'}$.
Note that a minimum vertex ordering may not exist.

The {\em{width}} of a vertex ordering $\pi$ of a digraph $D$, denoted $\cw(D,\pi)$, is equal to $\max\{\cutvector{D}{\pi}\}$, where $\max$ on a tuple yields the largest coordinate. 
The {\em{cutwidth}} of $D$, denoted $\cw(D)$, is the minimum width among vertex orderings of $D$.
Similarly, the {\em{cost}} of $\pi$, denoted $\ola(D,\pi)$, is equal to $\sum\{\cutvector{D}{\pi}\}$, where $\sum$ on a tuple yields the sum of coordinates.
This is equivalent to summing $|i-j|$ for all feedback arcs $(\onesmt{\pi}{j},\onesmt{\pi}{i})$ in $\pi$.
The {\em{OLA-cost}} of $D$, denoted $\ola(D)$, is the minimum cost among vertex orderings of $D$. 
A vertex ordering $\pi$ of $D$ satisfying $\cw(D) = \cw(D,\pi)$, or $\ola(D) = \ola(D,\pi)$, is respectively called \emph{$\cw$-optimal} or \emph{\ola-optimal} for $D$. 
Note that a minimum ordering for $D$, if existent, is always \emph{$\cw$-optimal} and \emph{\ola-optimal} for $D$. 

\paragraph{Kernelization.}
A {\em{kernelization algorithm}} (or {\em{kernel}}, for short) is a polynomial-time algorithm that given some instance of a parameterized problem, returns an equivalent instance whose size is bounded by a computable function of the input parameter; this function is called the {\em{size}} of the kernel.
We are mostly interested in finding polynomial kernels, as admitting a kernel of any computable size is equivalent to fixed-parameter tractability of the problem~\cite{DowneyF13,CFK14}. 

Based on the complexity hypothesis that $\mathsf{NP}\not\subseteq \mathsf{coNP}/\textrm{poly}$, various parameterized reduction and composition techniques allow to rule out the existence of a polynomial-size kernel. We refer to e.g.~\cite{CFK14} for an overview of this methodology. 
The concept of {\em{$\mathsf{AND}$-composition}} is one of these techniques. For a parameterized problem $\Pi$ (see~\cite{DowneyF13,CFK14} for basic definitions of parameterized complexity), an 
{\em{$\mathsf{AND}$-composition}} is a polynomial-time algorithm that a sequence of instances $(I_1,k),(I_2,k),\ldots, (I_t,k)$, all with the same parameter $k$,
outputs a single instance $(I^\star,k^\star)$, such that the following holds:
\begin{itemize}
\item $(I^\star,k^\star)\in \Pi$ if and only if $(I_i,k)\in \Pi$ for all $1\leq i\leq t$; and
\item $k^\star\leq q(k)$ for some polynomial $q$.
\end{itemize}

\begin{theorem}[\cite{Dru15}]\label{th:and}
Let $\Pi$ be an $\mathsf{NP}$-hard parameterized problem for which there exists an $\mathsf{AND}$-composition. 
Then, unless $\mathsf{NP}\subseteq \mathsf{coNP}/\textrm{poly}$, $\Pi$ does not admit a polynomial kernel.
\end{theorem}

The notion of {\em{Turing kernelization}} has been introduced as an alternative to classic kernelization (cf.
the discussion in~\cite{CFK14}).
The idea is to relax the notion of efficient pre-processing to allow producing multiple instances of small size, so that we can obtain positive results for problems which do not admit polynomial-size kernels in the strict sense. In this framework, it is not required that the instance at hand is reduced to one equivalent small instance, 
but rather that the whole problem can be solved in polynomial time assuming oracle access to an algorithm solving instances of size bounded by a function of the parameter.
More precisely, a Turing kernel of size $q(k)$ for a parameterized problem $\Pi$ is an algorithm that, given an instance $(I,k)$, resolves whether $(I,k)\in \Pi$ in polynomial time when given access to
an oracle that resolves belonging to $\Pi$ for instances of size at most $q(k)$. Each oracle call is counted as a single step of the algorithm. As with classic kernelization, we typically assume that $q(k)$ is computable,
and in practice we are looking for {\em{polynomial Turing kernels}} where $q(k)$ is a polynomial.

\paragraph{Exponential-Time Hypothesis.} 
The \emph{Exponential Time Hypothesis} (ETH) of Impagliazzo et al.~\cite{IPZ01} states that for some constant \blue{$\alpha > 0$}, there is no algorithm for {\sc{3SAT}} that would run in time $2^{\blue{\alpha}\cdot n} \cdot (n + m)^{\Oh(1)}$, 
where $n$ and $m$ are the numbers of variables and clauses of the input formula, respectively. 
Using the {\em{Sparsification Lemma}}~\cite{IPZ01} one can show that under ETH, there is a constant $\blue{\alpha}>0$ such that {\sc{3SAT}} cannot be solved in time $2^{\blue{\alpha} \cdot m} \cdot (n + m)^{\Oh(1)}$.
In this work we use the {\sc{NAE-3SAT}} problem (for Not-All-Equal), which is a variant of {\sc{3SAT}} where a clause is considered satisfied only when at least one, but not all of its literals are satisfied.
Schaefer~\cite{Sch78} gave a linear reduction from {\sc{3SAT}} to {\sc{NAE-3SAT}}, which immediately yields:

\begin{corollary}\label{cor:ETH} 
Unless ETH fails, {\sc{NAE-3SAT}} cannot be solved in time $2^{o(m)}\cdot (n+m)^{\Oh(1)}$, where $n$ and $m$ are the numbers of variables and clauses of the input formula, respectively. 
\end{corollary}


\section{Polynomial-time and approximation algorithms} 
\label{sec:polynomial}


The starting point of our study is the approach used in the earlier works by Fradkin~\cite{Fra11} and by the third author~\cite{Pil13,Pilipczuk13},
namely to sort the vertices of the given semi-complete digraph according to non-decreasing indegrees, and argue that this ordering has to resemble an optimum one. As shown by Fradkin~\cite{Fra11}, this statement may be made precise for tournaments: any indegree ordering has optimum cutwidth. A slight modification enabled Fradkin~\cite{Fra11} to establish a polynomial-time $2$-approximation for semi-complete digraphs. We choose to include the proofs of these two results (see Theorems~\ref{thm:poly-time} and~\ref{thm:apx} below) for several reasons. First, the applicability of the approach to \ola{} is a new contribution of this work. Moreover, the fine understanding of optimum orderings is a basic tool needed in the proofs of our main results. Finally, the abovementioned results of Fradkin~\cite{Fra11} were communicated only in her PhD  thesis and, to the best of our knowledge, were neither included in any published work, nor we have found any reference to them. We believe that these fundamental observations deserve a better publicity.

The core idea is to work in the more general setting of {\em{fractional tournaments}}, a linear relaxations of tournaments. Formally, a {\em{fractional tournament}} is a pair $T=(V,\omega)$, where $V$ is a finite vertex set and $\omega\colon V^2\to \mathbb{R}_{\geq 0}$ is a weight function that satisfies the following properties: 
$\omega(u,u)=0$ for all $u\in V$, and $\omega(u,v)+\omega(v,u)=1$ for all pairs of different vertices $u,v$. 
Thus, by requiring the weights to be integral we recover the original definition of a tournament. 
We extend the notation for digraphs to fractional tournaments as follows. For $X,Y\subseteq V$ we define 
$\omega(X,Y)=\sum_{x\in X,\, y\in Y} \omega(x,y)$, and for $u\in V$ we define $\omega^-(u)=\omega(V,\{u\})$ and $\omega^+(u)=\omega(\{u\},V)$.
The notions of (minimum) vertex orderings, cut vectors, cutwidth, and OLA-cost are extended naturally: the cardinality of any cut $E(X,Y)$ is replaced by the sum of weights $\omega(X,Y)$.

Suppose $T=(V,\omega)$ is a fractional tournament. We say that a vertex ordering $\pi$ of $T$ is {\em{sorted}} if for any pair of different vertices $u$ and $v$,
if $\omega^-(u)<\omega^-(v)$, then $\pi(u)<\pi(v)$; in other words, the vertices are sorted according to their indegrees. The following lemma extends the sorting arguments of Fradkin~\cite{Fra11}, and
encapsulates the essence of our approach.

\begin{lemma}\label{lem:minimum-sorted}
A vertex ordering of a fractional tournament is minimum iff it is sorted.
\end{lemma}
\begin{proof}
Let $T=(V,\omega)$ be a fractional tournament. By definition, a vertex ordering $\sigma$ of $T$ is minimum if and only if for each vertex ordering $\pi$ and each $i \in [0,|V|]$, the following inequality~($\diamondsuit$) holds: $\cutvector{T}{\sigma}(i)\leq \cutvector{T}{\pi}(i)$. Observe that, since $T$ is a fractional tournament, for any vertex ordering $\pi$ of $T$ and position $i \in [0,|V|]$ we have that\begin{gather*}
\cutvector{T}{\pi}(i) = \omega(\lstsmt{\pi}{i},\fstsmt{\pi}{i}) = \omega(V,\fstsmt{\pi}{i}) - \omega(\fstsmt{\pi}{i},\fstsmt{\pi}{i}) \text{, and}\\
\omega(\fstsmt{\pi}{i},\fstsmt{\pi}{i}) =\sum_{u\in \fstsmt{\pi}{i}} \omega(u,u)+\sum_{\substack{\{u,v\}\subseteq \fstsmt{\pi}{i},\\ u\neq v}} (\omega(u,v)+\omega(v,u))=\binom{i}{2}.
\end{gather*}
Hence, the inequality ($\diamondsuit$) is equivalent to $\omega(V,\fstsmt{\sigma}{i}) \leq \omega(V,\fstsmt{\pi}{i})$.
It is clear that the vertex ordering $\sigma$ minimizes $\omega(V,\fstsmt{\sigma}{i})$ only when $\fstsmt{\sigma}{i}$ consists of $i$ vertices of $T$ with the smallest indegrees $\omega^-$. But the latter condition holds for all $i \in [0,|V|]$ if and only if $\sigma$ is sorted.
\end{proof}

\begin{theorem}\label{thm:poly-time}
The cutwidth and \ola{} of a tournament can be computed in polynomial time.
\end{theorem}

The proof of Theorem~\ref{thm:poly-time}, even in the more general setting of fractional tournaments, is now immediate. We just sort the vertices according to their indegrees $\omega^-$. By Lemma~\ref{lem:minimum-sorted}, the obtained ordering is minimum, hence it is both $\cw$-optimal and \ola-optimal.

\begin{figure}[h]
\centering
\begin{tikzpicture} 
\tikzstyle{vertex}=[draw,circle,minimum size=0.5cm,inner sep=0pt]
\node[vertex] (u) at (-1.2,0) {$u$};
\node[vertex] (3) at (0,-0.6) {};
\node[vertex] (1) at (0,.6) {};
\node[vertex] (4) at (1.2,-.6) {};
\node[vertex] (2) at (1.2,.6) {};
\draw[<->,>=latex] (1) -- (2);
\draw[<->,>=latex] (1) -- (3);
\draw[<->,>=latex] (1) -- (4);
\draw[<->,>=latex] (2) -- (3);
\draw[<->,>=latex] (2) -- (4);
\draw[<->,>=latex] (3) -- (4);
\draw[->,>=latex] (u) -- (3);
\draw[->,>=latex] (u) -- (1);
\draw[->,>=latex] (4) to[out = -45,in = -135] (u);
\draw[->,>=latex] (2) to[out = 45,in = 135] (u);
\end{tikzpicture} \vspace{-5mm}
\caption{A semi-complete digraph without any minimum ordering. Since $u$ has both the minimum indegree and the minimum outdegree, it minimizes the first or last cut depending whether it is the first or last vertex of an ordering. 
}\label{fig:noMin}
\end{figure}
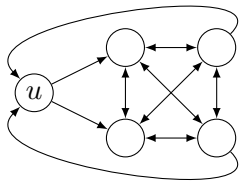

Observe that Lemma~\ref{lem:minimum-sorted} cannot be generalized to the semi-complete setting, as there are semi-complete digraphs that do not admit any minimum ordering; see Figure~\ref{fig:noMin} for an example. However, relaxing a given semi-complete digraph to a fractional tournament yields a $2$-approximation algorithm for general semi-complete digraphs.

Precisely, for a semi-complete digraph $D$, consider its {\em{relaxation}} $T_D$ which is a fractional tournament on the vertex set $V(D)$, where for every pair of different vertices $u$ and $v$, we put:
\begin{itemize}[nosep]
\item $\omega(u,v)=1$ and $\omega(v,u)=0$, when $(u,v)$ is present in $D$ but $(v,u)$ is not present; and
\item $\omega(u,v)=\omega(v,u)=1/2$, when $(u,v)$ and $(v,u)$ is a pair of symmetric arcs in $D$.
\end{itemize}
We put $\omega(u,u)=0$ for every vertex $u$, thus $T_D$ is indeed a fractional tournament.
Observe that for any pair of vertices $u,v$, we have $|E(\{u\},\{v\})|/2\leq \omega_{T_D}(\{u\},\{v\})\leq |E(\{u\},\{v\})|$.
Therefore, for every vertex ordering $\pi$ of $D$ and every index $i\in [0,n]$, it holds that
\begin{align*}
\cutvector{D}{\pi}(i)/2 \leq \cutvector{T_D}{\pi}(i)\leq \cutvector{D}{\pi}(i).
\end{align*}
In particular we have
$\cw(D)/2\leq \cw(T_D) \leq \cw(D)$ and $\ola(D)/2\leq \ola(T_D) \leq \ola(D)$. As a consequence, we have:

\begin{theorem}\label{thm:apx}
There exists a polynomial-time algorithm that given a semi-complete digraph $D$, 
outputs an ordering of its vertices of width, and respectively cost, upper bounded by twice the cutwidth, respectively \ola{}-cost, of $D$.
\end{theorem}
\begin{proof}
Given a semi-complete digraph~$D$, we compute its relaxation $T_D$ and we sort the vertices according to their indegrees in $T_D$,
obtaining a vertex ordering $\pi$.
Lemma~\ref{lem:minimum-sorted} ensures us that 
\begin{align*}
\cw(T_D,\pi)=\cw(T_D)\qquad\textrm{and}\qquad \ola(T_D,\pi)=\ola(T_D).
\end{align*}
On the other hand, by inequalities 
\begin{align*}
\cw(D)/2\leq \cw(T_D) \leq \cw(D)\qquad\textrm{and}\qquad\ola(D)/2\leq \ola(T_D) \leq \ola(D)
\end{align*}
we have that
\begin{align*}
\cw(D,\pi)/2\leq \cw(T_D,\pi),\qquad & \ola(D,\pi)/2\leq \ola(T_D,\pi),\\
\cw(T_D)\leq \cw(D),\qquad & \ola(T_D)\leq \ola(D).
\end{align*}
By combining these together, we obtain the required upper bounds:
\begin{align*}
\cw(D,\pi)\leq 2\cw(D)\qquad\textrm{and}\qquad \ola(D,\pi)\leq 2\ola(D),
\end{align*}
which means that the ordering $\pi$ can be reported by the algorithm.
\end{proof}

Finally, we remark that Theorem~\ref{thm:apx} can be generalized in several directions. 
First, here we considered parameters cutwidth and \ola-cost, which are defined by taking the maximum and the sum over the cut vector
corresponding to an ordering. However, we used only a few basic properties of the max and sum functions; the whole reasoning would go through for any function over the cut vectors that is monotone
with respect to the $\preceq$ order, and scales by $\alpha$ when the cut vector is multiplied by $\alpha$.
Second, instead of semi-complete digraphs, we could consider weighted semi-complete digraphs defined as follows: we have a weight function $\omega\colon V^2\to \mathbb{R}_{\geq 0}$,
where we require that $\omega(u,u)=0$ for every vertex $u$, and $\omega(u,v)+\omega(v,u)>0$ for every pair of different vertices $u$ and $v$.
Such a weighted semi-complete digraph can be rescaled to a fractional tournament by taking the normalized weight function $\omega'(u,v)=\frac{\omega(u,v)}{\omega(u,v)+\omega(v,u)}$.
Then, the same reasoning as above yields an approximation algorithm for the cutwidth and the \ola-cost with approximation factor $\frac{\omega^{\max}_D}{\omega^{\min}_D}$, where
\begin{align*}
\omega^{\max}_D = \underset{\substack{u, v \in V(D)\\u \neq v}}{\max} \; (\omega(u,v) + \omega(v,u))\qquad \text{and}\qquad  \omega^{\min}_D = \underset{\substack{u, v \in V(D)\\u \neq v}}{\min} \; (\omega(u,v) + \omega(v,u)).
\end{align*}

\newcommand\redd[1]{#1}

\section{Kernelization aspects of {\sc{Cutwidth}} and {\sc{Ola}} }
\label{sec:turing-kernel}

\subsection{On polynomial kernels}

\blue{As announced in the introduction, the {\sc{Cutwidth}} and {\sc{Ola}} problems, when respectively parameterized by the target width and target cost, 
behave differently with respect to the kernelization paradigm. As we shall see next, the problem of computing the \ola{}-cost admits a linear kernel in arbitrary digraphs, 
while computing the {\sc{Cutwidth}} of semi-complete digraphs is unlikely to admit a polynomial size kernel. 
The proofs of these two results, stated as Theorems~\ref{thm:ola-kernel} and~\ref{thm:no-kernel} below, directly follow from the understanding of the contribution of strongly connected components in optimal orderings.
On one side, the contribution to the \ola{}-cost of each strongly connected component is at least linear in its size, implying a linear kernel for this parameter.
On the other side, we can observe that the cutwidth of digraph is the maximum over the cutwidth of its strongly connected components, 
which implies that, like many other width parameters, cutwidth is an \emph{$\mathsf{AND}$-composable} parameter~\cite{Dru15,CFK14}. 
}


\begin{theorem} \label{thm:ola-kernel}
The {\sc{Ola}} problem in \blue{arbitrary} digraphs, parameterized by the target cost $k$, admits a kernel with at most $2k$ vertices.
\end{theorem}
\begin{proof}
\blue{
Let $D$ be a digraph. As the set of strongly connected components of a digraph naturally defines an acyclic digraph, we can choose an ordering $C_1,\ldots,C_p$ of the strongly connected components so that whenever $u\in C_i$ and $v\in C_j$ for some $i<j$,
then the arc $(v,u)$ is not present in $D$ (such an ordering is given by a topological ordering of the DAG of the strongly connected components of $D$).
}

Observe that since $C_1,\ldots,C_p$ are pairwise disjoint, we have
\begin{equation*}
\ola(D)\geq \ola(D[C_1])+\ldots+\ola(D[C_p]).
\end{equation*}
On the other hand, there exists a vertex ordering of $D$ of cost equal to $\ola(D[C_1])+\ldots+\ola(D[C_p])$: one may simply concatenate optimum-cost orderings of $D[C_1], D[C_2],\ldots, D[C_p]$.
This means that in fact
\begin{equation}\label{eq:ola-sum}
\ola(D)=\ola(D[C_1])+\ldots+\ola(D[C_p]).
\end{equation}
Now observe that if any strongly connected component $C_i$ has only one vertex, then its removal from the digraph does not change the OLA-cost. 
Let then $D'$ be obtained from $D$ by removing all the strongly connected components of size one; then $\ola(D')=\ola(D)$.
Clearly, if $D'$ has at most $2k$ vertices, then it can be reported by the algorithm.
We argue that otherwise, if $D'$ has more than $2k$ vertices, then the algorithm can conclude that $\ola(D)=\ola(D')>k$.

Take any connected component $C_i$, and examine $D[C_i]$. Since $D[C_i]$ is strongly connected, for any its vertex ordering $\pi$, all the cuts 
$E(\lstsmt{\pi}{j},\fstsmt{\pi}{j})$ for $j\in \{1,2,\ldots,|C_i|-1\}$ will be nonempty. Consequently, the cost of any ordering $\pi$ will be at least $|C_i|-1$, so
\begin{equation}\label{eq:sub1}
\ola(C_i)\geq |C_i|-1.
\end{equation}
Observe that if $|C_i|>1$ then $|C_i|-1\geq |C_i|/2$. Hence, by~\eqref{eq:ola-sum} and~\eqref{eq:sub1} we have
\begin{equation*}
\ola(D)= \sum_{i=1}^p \ola(D[C_i])\geq \sum_{i=1}^p (|C_i|-1)\geq \frac{1}{2}\,\sum_{i\colon |C_i|>1}\, |C_i|=\frac{1}{2}|V(D')|.
\end{equation*}
We conclude that if $D'$ has more than $2k$ vertices, then $\ola(D)>k$.  
\end{proof}

\blue{Let us now turn to the {\sc{Cutwidth}} problem in semi-complete digraphs. We show that this problem admits a simple  $\mathsf{AND}$-composition. }

\begin{theorem}\label{thm:no-kernel}
Unless $\mathsf{NP}\subseteq \mathsf{coNP}/\textrm{poly}$, there exists no polynomial-size kernelization algorithm for {\sc{Cutwidth}} problem in semi-complete digraphs when parameterized by the target width.
\end{theorem}
\begin{proof}
We give an $\mathsf{AND}$-composition algorithm for the problem.
\blue{Let $(D_1,c), (D_2,c),\dots, (D_t,c)$ be the input instances, where each digraph $D_i$ is semi-complete.} The algorithm constructs a semi-complete digraph $D^\star$ by taking the disjoint union of $D_1,D_2,\ldots,D_t$, and for all $1\leq i<j\leq t$, adding arcs $(u,v)$ for all $u\in V(D_i)$ and $v\in V(D_j)$. The output instance is $(D^\star,c)$, so it remains to argue that $\cw(D^\star)\leq c$ if and only if $\cw(D_i)\leq c$ for all $i=1,2,\ldots,t$.
In one direction, if each $D_i$ admits a vertex ordering $\pi_i$ of width at most $c$, then the concatenation 
\blue{$\pi_1\cdot \pi_2\cdot \ldots\cdot \pi_t$} of these orderings is a vertex ordering of $D^\star$ of width at most $c$.
Conversely, each $D_i$ is an induced subdigraph of $D^\star$, so if $\cw(D_i)>c$ then also $\cw(D^\star)>c$.
We conclude that indeed we have $\cw(D^\star)\leq c$ if and only if $\cw(D_i)\leq c$ for all $i=1,2,\ldots,t$.

Now the statement directly follows from Theorem~\ref{th:and} and the fact that computing the cutwidth of a semi-complete digraph is an $\mathsf{NP}$-complete problem (see Theorem~\ref{thm:lower-bound} in Section~\ref{sec:lower-bound}).
\end{proof}

\subsection{A Turing kernel for the {\sc{Cutwidth}} problem}


We now prove that an instance of the {\sc{Cutwidth}} problem can be reduced to polynomially many instances of quadratic size which altogether are equivalent to the original instance; this yields
a quadratic Turing kernel for the problem, as announced in Section~\ref{sec:intro}. 
The following technical statement, which we will prove next, describes the output of this procedure in details.

\begin{proposition}\label{prop:turing-kernel}
There exists a polynomial-time algorithm that given a semi-complete digraph $D$ and integer $c$, either correctly concludes that $\cw(D)>c$ or 
outputs a list of at most $n$ induced subdigraphs $D_1,\ldots,D_{\ell}$ of $D$, each with at most \redd{$24c^2+40c+1$} vertices, 
such that $\cw(D) \leq c$ if and only if $\cw(D_i)\leq c$ for each $i\in \{1,2,\ldots,\ell\}$.
\end{proposition}

Proposition~\ref{prop:turing-kernel} immediately implies that the problem admits a Turing kernel. 

\begin{theorem}\label{thm:turing-kernel}
The \textsc{Cutwidth} problem in semi-complete digraphs, parameterized by the target width $c$, admits a Turing kernel with $\Oh(n)$ calls to an oracle solving instances of size $\Oh(c^2)$.
\end{theorem}
\begin{proof}
Apply the algorithm of Proposition~\ref{prop:turing-kernel} and use an oracle call for each output induced subdigraph $D_i$ to resolve whether its cutwidth is at most $c$.
\end{proof}

The remainder of this section is devoted to the proof of Proposition~\ref{prop:turing-kernel}.
\blue{The line of the reasoning is as follows. We first compute a $2$-approximate ordering using Theorem~\ref{thm:apx}, and then we exhaustively improve it
until it becomes a lean ordering (see definition below)}. Let $\sigma$ be the obtained ordering, and consider the sequence of cutsizes along $\sigma$. The next observation is crucial. Due to leanness, if some cutsize in this sequence 
is smaller or equal than $\Omega(c)$ cutsizes to the left and to the right, then there is some optimum-width ordering that uses the corresponding cut; that is, the prefix of $\sigma$ up to this cut is also 
a prefix of some optimum-width ordering. We call such cuts {\em{milestones}}. It is not hard to prove that a milestone can be found every $\Oh(c^2)$ vertices in the ordering $\sigma$. 
Thus we are able to partition the digraph
into pieces of size $\Oh(c^2)$ that may be treated independently. Each of these pieces gives rise to one digraph $D_i$ in the output of the kernelization algorithm.

Before proceeding to the proof itself, we first prove an auxiliary statement that intuitively says the following:
if two vertices are far from each another in a given ordering of width at most~$c$, then their relative ordering remains unchanged in any other ordering of width at most~$c$.

\begin{lemma}\label{lem:dispersed}
Let $D=(V,E)$ be a semi-complete digraph and let $\pi$ be its vertex ordering of width at most $c$. Suppose $u,v\in V$ are such that $\pi(u)<\pi(v)-3c$.
Then, for every vertex ordering $\sigma$ of $D$ of width at most $c$, we have that $u$ is placed before $v$ in $\sigma$.
\end{lemma}
\begin{proof}
Let $X$ be the set of vertices placed between $u$ and $v$ in $\pi$ (exclusive). As $\pi(u)< \pi(v)-3c$, we have $|X|\geq 3c$.
We are going to exhibit $c+1$ arc-disjoint paths leading from $u$ to $v$. This implies that in any ordering $\sigma$ of width at most $c$, $u$ must be placed before $v$, as
otherwise each of these paths would contribute to the size of every cut in $\sigma$ between $v$ and $u$.

Suppose first that $(v,u)\in E$.
Observe that $|X\cap N^-(u)|\leq c-1$, since each of arc with tail in $X\cap N^-(u)$ and head being $u$ contributes to the size of the cut $E_\pi^{\pi(u)}$, and we moreover have the arc $(v,u)\in E$.
Similarly, $|X\cap N^+(v)|\leq c-1$. If we now denote $Y=X\setminus (N^-(u)\cup N^+(v))$, then we infer that $|Y|\geq 3c-2(c-1)=c+2$.
Since $D$ is semi-complete, we have that for each $w\in Y$ there are arcs $(u,w)$ and $(w,v)$, which form a path of length $2$ from $u$ to $v$. This is a family of $c+2$ arc-disjoint paths from $u$ to $v$.

Suppose now that $(v,u)\notin E$, hence $(u,v)\in E$ since $D$ is semi-complete.
Exactly as above we infer that $|X\cap N^-(u)|\leq c$ and $|X\cap N^+(v)|\leq c$ (we now do not have the additional arc $(v,u)$), which implies that $|Y|\geq 3c-2c=c$.
Again, vertices of $Y$ give rise to a family of $c$ arc-disjoint paths of length $2$ from $u$ to $v$, to which we can add the one-arc path consisting of $(u,v)$ to get $c+1$ arc-disjoint paths from $u$ to $v$.
\end{proof}

The essence of our approach is encapsulated in Lemma~\ref{lem:save-milestone} below, which provides a sufficient condition for a cut in a given ordering $\pi$ so that it can be assumed to be used in an optimum ordering $\sigma$.

\begin{lemma}\label{lem:save-milestone}
Let $D=(V,E)$ be a semi-complete digraph. Let $\pi$ and $\sigma$ be two vertex orderings of $D$ such that $\cw(D,\sigma) \leq \cw(D,\pi)=c$.
Suppose further that \redd{$m\in [3c,|V|-3c]$} is such that in $D$ there is a family of $|E_\pi^m|$ arc-disjoint paths leading from \redd{$\lstsmt{\pi}{m+3c}$} to \redd{$\fstsmt{\pi}{m-3c}$}. 
Then there exists a vertex ordering $\sigma^*$ such that:
\begin{itemize}[nosep]
\item $\fstsmt{\sigma^*}{m}=\fstsmt{\pi}{m}$;
\item for every $j$ with \redd{$j\leq m-3c$} or \redd{$j>m+3c$}, we have $\sigma^*_j=\sigma_j$;
\item $\cw(D,\sigma^*)\leq \cw(D,\sigma)$.
\end{itemize}
\end{lemma}

The intuition behind Lemma~\ref{lem:save-milestone} is as follows.
Consider $\sigma^*$ as rearranged $\sigma$. The second condition says that the rearrangement is local: it affects
only vertices at positions in the range \redd{$[m-3c+1,m+3c]$}. The third condition says that the rearrangement does not increase the width. Finally, the first condition is crucial: $\sigma^*$ uses
the prefix $\fstsmt{\pi}{m}$ of $\pi$ as one of its prefixes. Thus, the intuition is that any ordering can be locally rearranged while preserving the width so that prefix $\fstsmt{\pi}{m}$ is used,
provided there is a large arc-disjoint flow locally near $m$. See Figure~\ref{fig:rearrangement}~below.

\usetikzlibrary{shapes}
\begin{figure}[htbp]
\centering
\begin{tikzpicture}[scale=1.0]
   \tikzstyle{vertex1c}=[fill,circle,minimum size=0.2cm,inner sep=0pt]
   \tikzstyle{vertex2c}=[draw,circle,minimum size=0.2cm,inner sep=0pt]
   \tikzstyle{vertex2d}=[draw,diamond,minimum size=0.22cm,inner sep=0pt]
   \tikzstyle{vertex1d}=[fill,diamond,minimum size=0.22cm,inner sep=0pt]
   
   \foreach \x/\y in {c/-,d/+} {
      \foreach \i [evaluate=\i as \j using \y1*(\i*0.4-0.15)] [evaluate=\i as \k using \y1*(\i*0.4+2.35)] in {1,2,...,6} {
   \node[vertex1\x] (p\x{}b\i) at (\k,1.5) {};
   \node[vertex2\x] (p\x{}w\i) at (\j,1.5) {};
   }}
   
   \foreach \i/\poscb/\posdb/\posscw/\possdw/\postcw/\postdw in {
   1/4.75/2.25/-4.35/-1.05/4.35/0.25,
   2/3.95/2.75/-2.75/-0.65/2.75/0.65,
   3/3.55/3.55/-1.85/0.25/1.85/1.05,
   4/3.15/3.95/-0.25/1.05/1.05/1.45,
   5/2.25/4.35/0.65/1.85/0.65/1.85,
   6/1.45/4.75/1.45/3.15/0.25/3.15
   } 
{
   \node[vertex1c] (scb\i) at (-\poscb,.3) {};
   \node[vertex1c] (tcb\i) at (-\poscb,-.9) {};
   \node[vertex1d] (sdb\i) at (\posdb,.3) {};
   \node[vertex1d] (tdb\i) at (\posdb,-.9) {};
   \node[vertex2c] (scw\i) at (\posscw,.3) {};
   \node[vertex2d] (sdw\i) at (\possdw,.3) {};
   \node[vertex2c] (tcw\i) at (-\postcw,-.9) {};
   \node[vertex2d] (tdw\i) at (\postdw,-.9) {};
}
  
\foreach \i in {1,2,...,6} {
 \foreach \x in {db,dw,cb,cw} {
   \draw[gray!60,thick] (s\x\i) -- (t\x\i);
  }
}

   \foreach \h/\name/\indice in {1/$\pi$/$m+3c$,0/$\sigma$/$m$,-1/$\sigma^*$/$m-3c$} {
   \draw[thin,->,arrows={-latex'}] (-6,\h*1.2) -- (6,\h*1.2);
   \node[anchor=west] at (6,\h*1.2-0.1) {\name};
   
   \node at (-5.4,\h*1.2+0.2) {\ldots};
   \node at (5.4,\h*1.2+0.2) {\ldots};
   \draw[loosely dashed] (\h*2.5,1.9) -- (\h*2.5,-1.5);
   \node at (\h*2.4,2.1) {\indice};
   }
   
   
   \draw[dashed] (0,1.9) -- (0,-1.5);
   
\end{tikzpicture}
\caption{Situation in the proof of Lemma~\ref{lem:save-milestone}. Vertex orderings $\pi$, $\sigma$, and $\sigma^*$ are depicted, where equal vertices in $\sigma$ and $\sigma^*$ are connected by a grey line.
Circle nodes belong to $\fstsmt{\pi}{m}$, and diamond nodes belong to $\lstsmt{\pi}{m}$. White nodes belong to \redd{$\lstsmt{\pi}{m-3c} \cap \fstsmt{\pi}{m+3c}$}. 
Claim~\ref{cl1} says that the black nodes cannot be swapped in $\sigma$ with a node of the complementary shape.\label{fig:rearrangement}}
\end{figure}
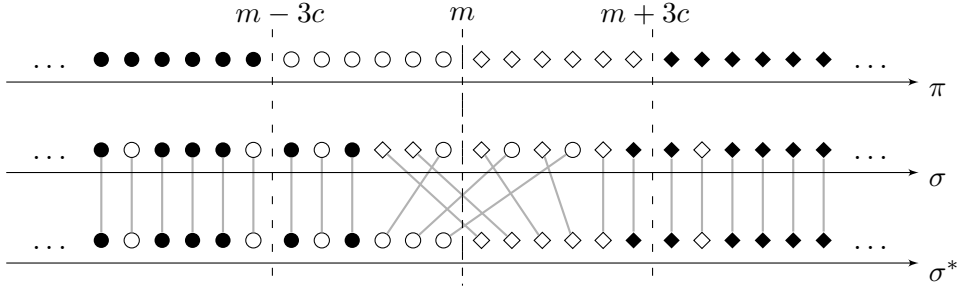

\begin{proof}[Proof of Lemma~\ref{lem:save-milestone}]
We first establish the following basic observation on the relation between orderings $\pi$ and $\sigma$. 

\begin{claim}\label{cl1}
In the ordering $\sigma$, every vertex of \redd{$\fstsmt{\pi}{m-3c}$} is placed before every vertex of $\lstsmt{\pi}{m}$,
and every vertex of $\fstsmt{\pi}{m}$ is placed before every vertex of \redd{$\lstsmt{\pi}{m+3c}$}.
\end{claim}
\begin{proof}
Consider any \redd{$u\in \fstsmt{\pi}{m-3c}$} and $v\in\lstsmt{\pi}{m}$. 
Then, $\pi(u) < \pi(v)-3c$, so the first statement follows immediately from Lemma~\ref{lem:dispersed}.
The proof of the second statement is the same.
\cqed
\end{proof}

Let $\sigma_{\leq}$ and $\sigma_>$ denote the restriction of $\sigma$ to $\fstsmt{\pi}{m}$ and $\lstsmt{\pi}{m}$, respectively.
Then, define $\sigma^*$ to be the concatenation of $\sigma_{\leq}$ and $\sigma_>$; see Figure~\ref{fig:rearrangement} for reference.
By the construction we have $\fstsmt{\pi}{m} = \fstsmt{\sigma^*}{m}$, so the first condition is satisfied.
For the second condition, observe that by Claim~\ref{cl1}, every vertex of \redd{$\fstsmt{\pi}{m-3c}$} is before every vertex of $\lstsmt{\pi}{m}$ in $\sigma$.
It follows that in $\sigma$, the first vertex of $\lstsmt{\pi}{m}$ appears only after a prefix of at least \redd{$m-3c$} vertices of $\fstsmt{\pi}{m}$. 
In the construction of $\sigma^*$ from $\sigma$, 
the vertices of that prefix stay at their original positions, so $\onesmt{\sigma^*}{j}=\onesmt{\sigma}{j}$ for all \redd{$j\leq m-3c$}.
A symmetric argument shows that $\onesmt{\sigma^*}{j}=\onesmt{\sigma}{j}$ also for all \redd{$j>m+3c$}.

It remains to prove that $\cw(D,\sigma^*)\leq \cw(D,\sigma)$.
Consider any $j\in [0,|V|]$; we need to prove that $|E_{\sigma^*}^j|\leq \cw(D,\sigma)$.
By the second condition we have that $E_{\sigma^*}^j=E_{\sigma}^j$ when  \redd{$j\leq m-3c$} or \redd{$j\geq m+3c$}, and $|E_{\sigma}^j|\leq \cw(D,\sigma)$ by definition.
Hence, we are left with checking the inequality for $j$ satisfying \redd{$m-3c<j< m+3c$}.

In the following, for a vertex subset $A$ we denote $\delta(A)=|E(V\setminus A,A)|$. We will use the submodularity of directed cuts:
$\delta(A\cap B)+\delta(A\cup B)\leq \delta(A)+\delta(B)$
for all vertex subsets $A,B$. In these terms, we need to prove that $\delta(\fstsmt{\sigma^*}{j})\leq \cw(D,\sigma)$.

Let $x$ be the vertex at position $j$ in $\sigma^*$ and let $X$ be the set containing $x$ and all vertices placed before $x$ in $\sigma$.
Suppose first that $j\leq m$. Then, by the construction we have $x\in \fstsmt{\pi}{m}$ and $\fstsmt{\sigma^*}{j}=X\cap \fstsmt{\pi}{m}$.
By the submodularity of cuts we have
\begin{equation}\label{eq:sub}
\blue{\delta(\fstsmt{\sigma^*}{j})=\delta(X\cap \fstsmt{\pi}{m})\leq \delta(X)+\delta(\fstsmt{\pi}{m})-\delta(X\cup \fstsmt{\pi}{m}).}
\end{equation}
As $X$ is a prefix of $\sigma$ by definition, we have $\delta(X)\leq \cw(D,\sigma)$.
\blue{Hence, by~\eqref{eq:sub}, in order to prove that  $\delta(\fstsmt{\sigma^*}{j})\leq \cw(D,\sigma)$, it suffices to prove that $\delta(X\cup \fstsmt{\pi}{m})\geq \delta(\fstsmt{\pi}{m})$.}

Denote $d=\delta(\fstsmt{\pi}{m})=|E_\pi^m|$ and recall that there is a family of $d$ arc-disjoint paths leading from \redd{$\lstsmt{\pi}{m+3c}$} to \redd{$\fstsmt{\pi}{m-3c}$}.
In particular, this means that for each set $A$ with \redd{$A\supseteq \fstsmt{\pi}{m-3c}$ and $A\cap \lstsmt{\pi}{m+3c}=\emptyset$}, each of these paths has to contribute to $\delta(A)$,
implying $\delta(A)\geq d$.

Therefore, it suffices to show that \redd{$X\cup \fstsmt{\pi}{m}\supseteq \fstsmt{\pi}{m-3c}$ and $(X\cup \fstsmt{\pi}{m})\cap \lstsmt{\pi}{m+3c}=\emptyset$}.
While the first assertion is trivial, the second is equivalent to \redd{$X\cap \lstsmt{\pi}{m+3c}=\emptyset$}.
For this, observe that by definition \blue{no vertex of $X$ is placed  after $x$ in $\sigma$}, 
and $x$ belongs to $\fstsmt{\pi}{m}$.
\blue{Moreover,}
by Claim~\ref{cl1} all vertices of \redd{$\lstsmt{\pi}{m+3c}$} are placed in $\sigma$ after all vertices of $\fstsmt{\pi}{m}$, in particular after $x$.
This implies that $X$ and \redd{$\lstsmt{\pi}{m+3c}$} are disjoint. By the discussion above, this proves that $\delta(X\cup \fstsmt{\pi}{m})\geq d$ and, consequently also that $\delta(\fstsmt{\sigma^*}{j})\leq \cw(D,\sigma)$.

The proof for the case $j>m$ is completely symmetric, however we need to observe that now $x\in \lstsmt{\pi}{m}$ and $\fstsmt{\sigma^*}{j}=X\cup \fstsmt{\pi}{m}$.
By applying the same submodularity argument~\eqref{eq:sub}, we are left with proving that $\delta(X\cap \fstsmt{\pi}{m})\geq \delta(\fstsmt{\pi}{m})$,
which follows by a symmetric reasoning.
\end{proof}

Our goal now is to construct an approximate ordering $\pi$ where we will be able to find many positions $m$ to which Lemma~\ref{lem:save-milestone} can be applied.
We first recall the concept of a lean ordering, which will be our main tool for finding families of arc-disjoint paths.

\begin{definition}
A vertex ordering $\pi$ of a digraph $D = (V,E)$ is called \emph{lean} if for each $0 \leq a \leq b \leq n$,
the maximum size of a family of arc-disjoint paths from $\lstsmt{\pi}{b}$ to $\fstsmt{\pi}{a}$ in $D$ is equal to $\min_{a\leq i\leq b} |E_\pi^i|$.
\end{definition}

Note that by Menger's theorem, the maximum size of a family of arc-disjoint paths from $\lstsmt{\pi}{b}$ to $\fstsmt{\pi}{a}$ is equal to the minimum size of 
\blue{a cut} 
separating \blue{$\lstsmt{\pi}{b}$ from $\fstsmt{\pi}{a}$. }
Thus, in a lean ordering we have that the minimum cutsize between any disjoint prefix and suffix is actually realized by one of 
the cuts along the ordering.

The notion of a lean ordering is the cutwidth analogue of a {\em{lean decomposition}} in the treewidth setting, cf.~\cite{Tho90}.
An essentially equivalent notion of {\em{linked orderings}} was used by Chudnovsky and Seymour~\cite{CS11} in the context of immersions in tournaments.
Also, Giannopoulou et al.~\cite{GPR16} used this concept to study immersion obstructions for the cutwidth of undirected graphs.
A careful analysis of the arguments of~\cite{CS11,GPR16} yields the following.

\begin{lemma}[\cite{CS11,GPR16}]\label{lem:lean}
There is a polynomial-time algorithm that given a vertex ordering $\pi$ of a digraph $D$, 
computes a lean vertex ordering $\pi^*$ of $D$ satisfying $\cw(D,\pi^*) \leq \cw(D,\pi)$.
\end{lemma}
\begin{proof}
To prove the statement we may adapt the proof of Lemma 3.1 from~\cite{CS11}, or the proof of Lemma 13 from~\cite{GPR16}.

Suppose there exist two integers $a$ and $b$ certifying that $\pi$ is not a lean ordering.
That is, if we denote $d=\min_{i\in [a,b]} |E_\pi^i|$, then there is no family of $d$ arc-disjoint paths leading from $\lstsmt{\pi}{b}$ to $\fstsmt{\pi}{a}$.
By Menger's Theorem there exists a partition $(A,B)$ of $V(D)$ such that $\fstsmt{\pi}{a}\subseteq A$, $\lstsmt{\pi}{b}\subseteq B$, and $|E(B,A)|<d$. 
Such a partition can be identified in polynomial time using a min-cut/max-flow algorithm. 

Consider the vertex ordering $\pi'$ obtained by concatenating the restriction $\pi_A$ of $\pi$ to $A$ with the restriction of \blue{$\pi_B$ of $\pi$ to $B$; that is, $\pi'=\pi_A\cdot \pi_B$.} 
The following now holds:

\begin{claim}\label{cl:subm-cite}
For every $i'\in [0,|V|]$, if the vertex $\onesmt{\pi'}{i'}$ is at position $i$ in $\pi$, then we have that $|E_{\pi'}^{i'}|\leq |E_{\pi}^i|$, and moreover this
inequality is strict for $i'=|A|$.
\end{claim}

The proof of Claim~\ref{cl:subm-cite} follows easily from submodularity of cuts, see the proof of 3.1 in~\cite{CS11} and the proof of Lemma~13 in~\cite{GPR16}; this latter paper considers undirected setting,
but the submodularity argument holds in the directed setting in the same way.
From Claim~\ref{cl:subm-cite} it follows that $\cw(D,\pi')\leq \cw(D,\pi)$ and $\ola(D,\pi')< \ola(D,\pi)$ (note that Lemma~13 in~\cite{GPR16} uses exactly this potential as the minimization goal).

Therefore, the algorithm starts with the input ordering $\pi$ and as long as the leanness condition is not satisfied for some $a$ and $b$, it computes the refined ordering $\pi'$ as above
and sets $\pi:=\pi'$. Note that the width of the ordering cannot increase in this process, while the \ola{}-cost strictly decreases at each iteration. Since the \ola{}-cost is at most $|V|^3$ at the beginning,
we infer that the algorithm outputs some lean ordering after at most a cubic number of iterations.
\end{proof}

Next, we introduce the concept of a \emph{milestone}. Intuitively, a milestone is a position where Lemma~\ref{lem:save-milestone} can be applied, provided the ordering is lean.

\begin{definition}
Let $\pi$ be a vertex ordering of a digraph $D = (V,E)$, and let $\alpha$ be a positive integer.
An integer $m\in [0,|V|]$ is a \emph{$\pi$-milestone of $D$ of span $\alpha$} if $|E_\pi^m| \leq |E_\pi^i|$ for each integer $i$ with $m-\alpha \leq i\leq m+\alpha$.
\end{definition}

Note that if $\pi$ is lean and $m$ is a $\pi$-milestone of span $\alpha$, then $\min_{m-\alpha\leq i\leq m+\alpha} |E_\pi^i|=|E_\pi^m|$, hence there is a family of $|E^\pi_m|$ arc-disjoint paths leading from
$\lstsmt{\pi}{m+\alpha}$ to $\fstsmt{\pi}{m-\alpha}$. Thus, a $\pi$-milestone of span \redd{$3c$} satisfies the prerequisite of Lemma~\ref{lem:save-milestone} about the existence of arc-disjoint paths. 
We now observe that, in an ordering of small width, every large enough set of consecutive positions contains a milestone.

\begin{lemma} \label{lem:milestone}
Let $D=(V,E)$ be a digraph and let $\pi$ be a vertex ordering of $D$ of width at most $c$. 
Then for any integers $p\in[0,|V|]$ and $\alpha\geq 0$, there exists a $\pi$-milestone $m\in [p-\alpha\cdot c,p+\alpha\cdot c]$ of span $\alpha$.
\end{lemma}
\begin{proof}
We look for a milestone by means of an iterative procedure.
Initially set $m:=p$. 
While $m$ is not a $\pi$-milestone of span $\alpha$, there exists an integer $j\in [m-\alpha,m+\alpha]$ such that $|E_\pi^j| < |E_\pi^m|$.
Then continue the iteration setting $m:=j$.
Observe that with each iteration, the cutsize $|E_\pi^m|$ strictly decreases.
Since $\cw(D,\pi)\leq c$, it follows that the number of the number of iterations before finding a $\pi$-milestone of span $\alpha$ is at most $c$.
Each iteration replaces $m$ with a number differing from it by at most $\alpha$, hence the final $\pi$-milestone $m$ satisfies $|p-m|\leq \alpha\cdot c$, as required.
\end{proof}

With all the tools gathered, we can finish the proof of Proposition~\ref{prop:turing-kernel}.

\begin{proof}[Proof of Proposition~\ref{prop:turing-kernel}]
By Theorem~\ref{thm:apx}, we can compute in polynomial time a vertex ordering $\pi_0$ of $D$ such that \blue{$\cw(D,\pi_0) \leq 2 \cdot\cw(D)$}. 
If $\cw(D,\pi_0) > 2c$, we can conclude that $\cw(D)>c$ and report this answer, so let us assume that $\cw(D,\pi_0) \leq 2c$. 
By applying the algorithm of Lemma~\ref{lem:lean} to $\pi_0$, we can compute in polynomial time a lean ordering $\pi$ such that $\cw(D,\pi) \leq \cw(D,\pi_0)\leq 2c$.
In the following we assume w.l.o.g. that \redd{$|V|>12c$}, for otherwise we can output a list consisting only of $D$.

Call a set of $\pi$-milestones {\em{dispersed}} if these $\pi$-milestone pairwise differ by more than \redd{$12c$}.
Observe that $0$ and $|V|$ are always $\pi$-milestones, and they differ by more than \redd{$12c$}.
Starting from the set $\{0,|V|\}$, we compute an inclusion-wise maximal dispersed set $0=m_0 < m_1 < m_2 < \ldots < m_\ell =|V|$ of $\pi$-milestones of span \redd{$6c$}.
More precisely, whenever some $\pi$-milestone of span \redd{$6c$} can be added to the set without spoiling the dispersity requirement, we do it, until no further such milestone can be added.
Observe that then we have that \redd{$m_{i+1}-m_i \leq 24c^2+24c+1$} for each $i\in [1,\ell-1]$, for otherwise the range \redd{$[m_i+12c+1,m_{i+1}-12c-1]$} would contain more than \redd{$24c^2$} vertices, 
so by Lemma~\ref{lem:milestone} we would be able to find in it a $\pi$-milestone of span \redd{$6c$} that could be added to the constructed dispersed set.

Thus, $\pi$ is partitioned into $\ell$ blocks $B_1,\ldots,B_{\ell}$, each of length at most \redd{$24c^2+24c+1$}, 
such that the $j$-th block $B_j$ is equal to $\{\onesmt{\pi}{m_{j-1}+1},\onesmt{\pi}{m_{j-1}+2},\ldots, \onesmt{\pi}{m_{j}}\}$.
For each $j\in [1,\ell]$, let $A_j$ be defined as $B_j$ augmented with the following vertices:
\begin{itemize}[nosep]
\item vertices at positions in ranges \redd{$[\max(1,m_{j-1}-6c+1),m_{j-1}]$ and $[m_{j}+1,\min(|V|,m_{j}+6c)]$},
\item all heads of arcs from \redd{$E^{m_{j-1}-6c}_\pi$}, and all tails of arcs from \redd{$E^{m_j+6c}_\pi$}.
\end{itemize}
Since the width of $\pi$ is at most $2c$, we have that \redd{$|A_j|\leq |B_j|+16c\leq 24c^2+40c+1$}.

For $j\in [1,\ell]$, let us denote $D_j=D[A_j]$. 
To prove the theorem, it now suffices to show that $\cw(D)\leq c$ if and only if $\cw(D_j)\leq c$ for each $j\in [1,\ell]$.
The forward direction is trivial, since cutwidth is closed under taking induced subdigraphs.
Hence, we are left with showing that if $\cw(D_j)\leq c$ for each $j\in [1,\ell]$, then $\cw(D)\leq c$.

Take any $j\in [1,\ell-1]$. As $m_j$ is a $\pi$-milestone of span \redd{$6c$}, we have \redd{$\underset{i \in [m_j-6c, m_j+6c]}\min |E_\pi^i|=|E_\pi^{m_j}|$.}
Since $\pi$ is lean, there is a family $\mathcal{F}_j$ of $|E_\pi^{m_j}|$ arc-disjoint paths in $D$ leading from \redd{$\lstsmt{\pi}{m_j+6c}$ to $\fstsmt{\pi}{m_j-6c}$}.
We can assume w.l.o.g. that each internal (non-endpoint) vertex of each of these paths has position between \redd{$m_j+6c$ and $m_j-6c+1$} in $\pi$.
Hence, in particular, each path of $\mathcal{F}_j$ starts with an arc of \redd{$E_\pi^{m_{j}+6c}$} and ends with an arc of \redd{$E_\pi^{m_{j}-6c}$}.
This implies that for each $j\in [1,\ell]$, all the paths of $\mathcal{F}_j$ are entirely contained both in $D_{j}$ and in $D_{j+1}$.

\begin{figure}[h]
\centering
\begin{tikzpicture}
   \draw[line width=1pt,->,>=latex] (-6,0) -- (6,0);
   \node[anchor=west] at (6,-0.05) {$\pi$};

   \draw[loosely dashed] (-5.5,0.7) -- (-5.5,-1);
   \draw[loosely dashed] (-1,0.7) -- (-1,-1);
   \draw[loosely dashed] (2,0.7) -- (2,-1);
   \draw[loosely dashed] (5.5,0.7) -- (5.5,-1);
   \draw (-4,0.7) -- (-4,-1);
   \draw (0.5,0.7) -- (0.5,-1);
   \draw (4,0.7) -- (4,-1);

   \draw[<->] (-5.4,-.8) -- (-4.1,-.8) node[below,midway] {$6c$};
   \draw[<->] (-0.9,-.8) -- (0.4,-.8) node[below,midway] {$6c$};
   \draw[<->] (0.6,-.8) -- (1.9,-.8) node[below,midway] {$6c$};
   \draw[<->] (4.1,-.8) -- (5.4,-.8) node[below,midway] {$6c$};

   \node at (-3.5,0.6) {\small $m_{j-1}$};
   \node at (0.8,0.6) {\small $m_j$};
   \node at (4.5,0.6) {\small $m_{j+1}$};

   \draw [decorate,decoration={brace,amplitude=6pt,raise=4pt},yshift=0pt]
(-4,0.8) -- (0.5,0.8) node [black,midway,yshift=0.55cm] {$B_j$};
   \draw [decorate,decoration={brace,amplitude=6pt,raise=4pt},yshift=0pt]
(2,-1.1) -- (-5.5,-1.1) node [black,midway,yshift=-.65cm] {$A_j \setminus \{u,v\}$};
\draw [decorate,gray,decoration={brace,amplitude=6pt,raise=4pt},yshift=0pt]
(5.5,-1.25) -- (-1,-1.25) node [gray,midway,yshift=-.65cm] {$A_{j+1} \setminus \{x,y\}$};
   \node[fill] (1) at (-3,0) {};
   \node[fill] (2) at (-1.6,0) {};
   \node[fill] (3) at (-.8,0) {};
   \node[fill] (4) at (0.2,0) {};
   \node[fill] (5) at (1.6,0) {};
   \node[fill] (6) at (2.4,0) {};
   \node[fill] (7) at (4.8,0) {};
   \node at (2.4,-0.3) {$u$};
   \node at (4.8,-0.3) {$v$};
   \node at (-3,0.3) {$x$};
   \node at (-1.6,0.3) {$y$};
   \draw[thick,->,arrows={-latex'}] (7) to [bend right = 20] (5);
   \draw[thick,->,arrows={-latex'}] (5) to [bend right = 20] (2);
   \draw[thick,->,arrows={-latex'}] (6) to [bend left = 20] (3);
   \draw[thick,->,arrows={-latex'}] (3) to [bend left = 25] (4);
   \draw[thick,->,arrows={-latex'}] (4) to [bend left = 20] (1);
\end{tikzpicture}
\caption{\blue{The thick arcs represent the paths of $\mathcal{F}_j$ that are contained both in $D_{j}$ and in $D_{j+1}$.}}
\end{figure}
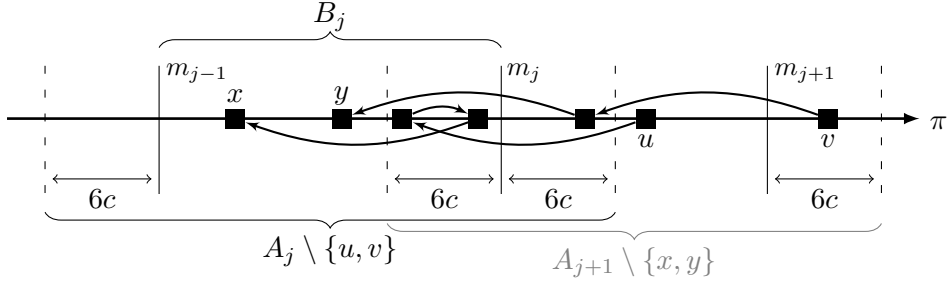

Consider any $j\in [1,\ell]$, and for simplicity assume for now that $j\neq 1$ and $j\neq \ell$.
Let $\pi'$ be the restriction of $\pi$ to the vertex set of $D_j$; obviously the width of $\pi'$ is at most $2c$.
Further, let $m'$ be the position of $\onesmt{\pi}{m_{j-1}}$ in $\pi'$, so that $\fstsmt{\pi'}{m'}=\fstsmt{\pi}{m_{j-1}}\cap V(D_j)$.
Observe that since all vertices at positions between \redd{$m_{j-1}-6c+1$} and \redd{$m_{j-1}+6c$} in $\pi$ are included in the vertex set of $D_j$,
they are at positions between \redd{$m'-6c+1$ and $m'+6c$} in $\pi'$, and hence the paths of $\mathcal{F}_{j-1}$ in $D_j$ lead from 
\redd{$\lstsmt{\pi'}{m'+6c}$ to $\fstsmt{\pi'}{m'-6c}$}. Their number is $|E^{m_{j-1}}_\pi|$, which is equal to the cutsize at position $m'$ in $\pi'$, by the construction of $D_j$ and $\pi'$.

We conclude that Lemma~\ref{lem:save-milestone} can be applied to position $m'$ in the ordering $\pi'$ of $D_j$.
If we now use it on any $\cw$-optimal vertex ordering $\sigma$ of $D_j$, we obtain a $\cw$-optimal vertex ordering $\sigma^*$ of $D_j$ such that
$\fstsmt{\sigma^*}{m'}=\blue{\fstsmt{\pi'}{m'}}=\fstsmt{\pi}{m_{j-1}}\cap V(D_j)$.
Note that by Lemma~\ref{lem:save-milestone}, $\sigma^*$ differs from $\sigma$ by a rearrangement of vertices at positions between \redd{$m'-6c+1$} and \redd{$m'+6c$}.

Now we define $m''$ to be the position of $\onesmt{\pi}{m_j}$ in $\pi'$, so that $\fstsmt{\pi'}{m''}=\fstsmt{\pi}{m_j}\cap V(D_j)$.
A symmetric reasoning, which uses the fact that $\mathcal{F}_j$ is also entirely contained in $D_j$, shows that Lemma~\ref{lem:save-milestone} can be also applied to position $m''$ in the ordering $\pi'$ of $D_j$.
Then we can use this lemma on the $\cw$-optimal vertex ordering $\sigma^*$, yielding a $\cw$-optimal ordering $\sigma^{**}$ such that 
$\fstsmt{\sigma^{**}}{m''}=\fstsmt{\pi'}{m''}=\fstsmt{\pi}{m_j}\cap V(D_j)$.
Again, by Lemma~\ref{lem:save-milestone} we have that $\sigma^*$ and $\sigma^{**}$ differ by a rearrangement of vertices at positions \redd{$m''-6c+1$ and $m''+6c$}.
Since \redd{$m_j-m_{j-1}>12c$} by construction, we infer that this rearrangement does not change the prefix of length $m'$, and hence we still have 
$\fstsmt{\sigma^{**}}{m'}=\fstsmt{\pi'}{m'}=\fstsmt{\pi}{m_{j-1}}\cap V(D_j)$.
The ordering $\sigma^{**}$ obtained in this manner shall be called $\sigma^j$.
For $j=1$ and $j=\ell$ we obtain $\sigma^j$ in exactly the same way, except we apply Lemma~\ref{lem:save-milestone} only once, for the position not placed at the end of the sequence.

All in all, for each $j\in [1,\ell]$ we have obtained a $\cw$-optimal ordering $\sigma^j$ of $D_j$ such that the vertices of $B_j$ form an infix (a sequence of consecutive elements) of $\sigma^j$,
while vertices to the left of this infix are the vertices of $V(D_j)\cap \fstsmt{\pi}{m_{j-1}}$ and vertices to the right of this infix are the vertices of $V(D_j)\cap \lstsmt{\pi}{m_j}$.
Define an ordering
$\tilde\sigma$ of $D$ by first restricting every ordering $\sigma^j$ to $B_j$, and then concatenating all the obtained orderings for $j=1,2,\ldots,\ell$.
Since we assumed that $\cw(D_j)\leq c$ for each $j\in [1,\ell]$, and each ordering $\sigma^j$ is $\cw$-optimal on $D_j$, we have that $\cw(D_j,\sigma^j)\leq c$ for each $j\in [1,\ell]$.
From the construction of $D_j$, and in particular the fact that all the arcs of $E_\pi^{m_{j-1}}$ and $E_\pi^{m_j}$ are contained in $D_j$, it follows that the infix of cutvector $\cutvector{D^j}{\sigma^j}$
corresponding to the vertices of $B_j$ is equal to the infix of the cutvector 
$\cutvector{D}{\blue{\tilde\sigma}}$ corresponding to the vertices of $B_j$. This shows that 
\begin{align*}
\cw(D,\blue{\tilde\sigma})=\max_{i\in [0,|V|]}\ \cutvector{D}{\blue{\tilde\sigma}}(i) \leq \underset{\substack{j \in [0,\ell]\\i\in [0,|V(D_j)|]}}{\max}\ \cutvector{D^j}{\sigma^j}(i) =\max_{j \in [0,\ell]}\ \cw(D_j,\sigma^j)\leq c,
\end{align*}
hence we are done.
\end{proof}


\section{Cutwidth-minimal semi-complete digraphs}
\label{sec:obstructions}

Recall here that a digraph $D$ is called \emph{$c$-cutwidth-minimal} if the cutwidth of $D$ is at least $c$,  but the cutwidth of every proper induced subdigraph of $D$ is smaller than $c$.
In this section we provide upper bounds on the sizes of $c$-cutwidth-minimal tournaments and semi-complete digraphs.
It turns out that the number of vertices in any $c$-cutwidth-minimal semi-complete digraph is bounded quadratically in $c$ (see Theorem~\ref{thm:obstructions-semi} below), while for
$c$-cutwidth-minimal tournaments we can even give an almost tight upper bound that is linear in $c$ (see Theorem~\ref{thm:obstructions-tour} below).
Essentially, the first result follows easily by considering applying the algorithm of Theorem~\ref{thm:turing-kernel} on a $c$-cutwidth-minimal semi-complete digraph for parameter $c-1$.
For the second result we use the understanding of minimum orderings in tournaments in the spirit of Lemma~\ref{lem:minimum-sorted}. 
Finally, we also discuss direct algorithmic applications of both these theorems.

\subsection{Upper bound for $c$-cutwidth-minimal tournaments}\label{sec:obst-tour}

\blue{We first provide a linear bound on the sizes of $c$-cutwidth-minimal tournaments.}
Our main tool will be the notion of a \emph{degree tangle}, introduced in~\cite{Pilipczuk13} as a certificate for large cutwidth.

\begin{definition}
For a digraph $D$ and nonnegative integers $k$ and $\alpha$, a {\em{$(k,\alpha)$-degree tangle}} is a vertex set $W \subseteq V(D)$  such that $|W| \geq k$, and for each $u,v \in W$ we have $|d^-(u) - d^-(v)| \leq \alpha$.
\end{definition}

In~\cite{Pilipczuk13,Pil13} it is essentially shown that if a semi-complete digraph $D$ admits a $(k,\alpha)$-degree tangle, then the cutwidth of $D$ is at least linear in $k-\alpha$.
In (fractional) tournaments we can establish a quadratic bound, as shown next.
\blue{Relying on the characterization of minimum orderings (Lemma~\ref{lem:minimum-sorted}), Lemma~\ref{lem:pdt} below provides a 
slightly finer understanding of the relation between degree tangles and cutwidth than~\cite{Pilipczuk13,Pil13}.}

\begin{lemma}\label{lem:pdt}
Let $k$ and $\alpha$ be nonnegative integers. If a fractional tournament $T$ contains a $(2k+1,\alpha)$-degree tangle, then $\cw(T) \geq \frac {k(k+1-\alpha)}2$.
\end{lemma}
\begin{proof}
Let $W$ be the $(2k+1,\alpha)$-degree tangle present in the fractional tournament $T$. Consider a sorted ordering $\pi$ of $T$. 
By Lemma \ref{lem:minimum-sorted}, we have $0 \leq \cutvector{T}{\pi}(i) \leq \cw(T)$ for any $i \in [0,n]$. Observe that
$$\cutvector{T}{\pi}(i) = \omega(V,\fstsmt{\pi}{i}) - \omega(\fstsmt{\pi}{i},\fstsmt{\pi}{i}) = \sum_{u \in \fstsmt{\pi}{i}} \omega^-(u) - \dbinom{i}{2} = \sum_{j = 1}^i (\omega^-(\onesmt{\pi}{j}) - (j-1)).$$
We infer that for any position $p \in [0,n-2k]$, the two following inequalities hold:
\begin{align}
\sum_{j = p}^{p+k-1} (\omega^-(\onesmt{\pi}{j})+1-j) = ~ \cutvector{T}{\pi}(\blue{p+k-1})- \cutvector{T}{\pi}(p-1) & \leq \cw(T)\label{eq:pos}\\
\sum_{j = p+k+1}^{p+2k} -(\omega^-(\onesmt{\pi}{j})+1-j) = ~ \cutvector{T}{\pi}(p+k)- \cutvector{T}{\pi}(p+2k) & \leq \cw(T)\label{eq:neg}
\end{align}

Set $p = \underset{v \in W}\min~\pi(v)$, that is, $p$ is the lowest position occupied by a vertex of $W$. Let $\delta = \omega^-(\pi_p)$. 
\blue{Since $\pi$ is sorted, without loss of generality we may assume that $W = \{\onesmt{\pi}{j} \colon j \in [p,p+2k]\}$.}
As $W$ is a $(2k+1,\alpha)$-degree tangle, for each $j \in [p,p+2k]$, we have $\delta \leq \omega^-(\pi_j) \leq \delta+\alpha$. 

Suppose $\delta+1 \geq p+k-\frac{\alpha}2$. We use~\eqref{eq:pos} together with $\omega^-(\pi_j) \geq \delta$ and re-indexing $i' = p+k-j$ to get: 
\begin{equation*}
\frac {k(k+1-\alpha)}2 = \sum_{i' = 1}^k (i'-\frac{\alpha}2) = \sum_{j = p}^{p+k-1} (p+k-j-\frac{\alpha}{2}) \leq \sum_{j = p}^{p+k-1} (\delta+1-j) \leq \cw(T).
\end{equation*}
Otherwise, we have $\delta+1 < p+k-\frac{\alpha}2$. Then we use~\eqref{eq:neg} together with $\omega^-(\onesmt{\pi}{j}) \leq \delta + \alpha$ and re-indexing $i' = j-p-k$ to get:
\begin{equation*}
\frac {k(k+1-\alpha)}2 = \sum_{i' = 1}^k (i'-\frac{\alpha}2) =  \sum_{j = p+k+1}^{p+2k} (j-p-k-\frac{\alpha}{2}) < \sum_{j = p+k+1}^{p+2k} -(\delta+\alpha+1-j) \leq \cw(T).
\end{equation*}
In both cases, we conclude that $\frac {k(k+1-\alpha)}2 \leq \cw(T)$.
\end{proof}

\blue{We are now ready to show that $c$-cutwidth-minimal tournaments have sizes linear in $c$.}

\begin{theorem}\label{thm:obstructions-tour}
For every $c \in \mathbb{N}$, every $c$-cutwidth-minimal tournament has at most $2c+2\lceil \sqrt{2c}\rceil+1$ vertices.
\end{theorem}
\begin{proof}
Consider a $c$-cutwidth-minimal tournament $T = (V,E)$ and let $\pi$ be a sorted ordering of $T$. By Lemma~\ref{lem:minimum-sorted}, we have $\cw(T,\pi)=\cw(T)\geq c$. 
Let $i\in [0,|V|]$ be such that $\cutvector{T}{\pi}(i) \geq c$. 
We define the vertex set $V_==\{v\in V \colon d^-(v)=d^-(\onesmt{\pi}{i})\}$. 
We let $V_>$ and $V_<$ denote the sets of vertices that respectively appear after and before $V_=$ in $\pi$.

Observe that for each $u \in V_=$, the set $V_= \setminus \{u\}$ is a $(|V_=|-1,1)$-degree tangle for $T[V \setminus \{u\}]$. 
By the minimality of $T$ and Lemma~\ref{lem:pdt}, we may thus assume that $|V_=|\leq 2\lceil \sqrt{2c} \rceil+1$, 
as otherwise the removal of any vertex of $V_=$ would still leave a degree tangle that certifies that the cutwidth is at least $c$.

We now focus on the set $V_>$. Note that $V_<\cup V_=$ is a prefix of $\pi$ and $V_>$ is a suffix of~$\pi$. Consider now removing any vertex $v$ of $V_>$ from the tournament $T$. This operation may decrease the indegrees of vertices of $V_>$ by at most one, so after the removal it will be still true that the indegree of any vertex of
$V_>$ will be at least as large as the indegree of any vertex of $V_<\cup V_=$. Consequently, there is a sorted vertex ordering of $T[V\setminus \{v\}]$ where $V_<\cup V_=$ is a prefix and $V_>\setminus \{v\}$ is
a suffix. Moreover, if $v$ had no outneighbors in $V_<\cup V_=$, then we could choose this sorted ordering so that on $V_<\cup V_=$ it would match $\pi$, implying that $E^i_\pi$ would be also a cut;
note here that no arc of $E^i_\pi$ is incident to $v$, as $v$ has no outneighbors in $V_<\cup V_=$.
By Lemma~\ref{lem:minimum-sorted}, this would mean that the cutwidth of $T[V\setminus \{v\}]$ would be at least $|E^i_\pi|$, which is at least $c$, a contradiction to the minimality of~$T$.
We conclude that each $v\in V_>$ has an outneighbor in $V_<\cup V_=$. 

Since there is a sorted ordering of $T[V\setminus \{v\}]$ where $V_<\cup V_=$ is a prefix and $V_>\setminus \{v\}$ is a suffix,
by Lemma~\ref{lem:minimum-sorted} we infer that the cutwidth of $T[V\setminus \{v\}]$ has to be at least $|E(V_>\setminus\{v\},V_<\cup V_=)|$. 
However, we have just argued that every vertex of $V_>$ has an outneighbor in $V_<\cup V_=$, so $|E(V_>\setminus\{v\},V_<\cup V_=)|\geq |V_>\setminus \{v\}|$. 
By the minimality of $T$, this implies that $|V_>\setminus \{v\}|\leq \cw(T[V\setminus \{v\}])<c$, so $|V_>|\leq c$. 
A symmetric argument shows that $|V_<|\leq c$ as well. The claim follows from combining the obtained upper bounds on the sizes of $V_<$, $V_=$, and $V_>$.
\end{proof}

Note that the bound of Theorem \ref{thm:obstructions-tour} is almost tight, as Figure~\ref{fig:minTournament} displays a $c$-cutwidth-minimal tournament with $2c + 1$ vertices;
we leave the easy verification of minimality to the reader. We also remark that the example in Figure~\ref{fig:minTournament} may be modified by replacing the depicted matching of backwards arcs
by any matching of backward arcs of size $c$ with tails at positions between $c+2$ and $2c+1$ and heads at positions between $1$ and $c$. This yields an exponential number of pairwise non-isomorphic
$c$-cutwidth-minimal tournaments on $2c+1$ vertices.


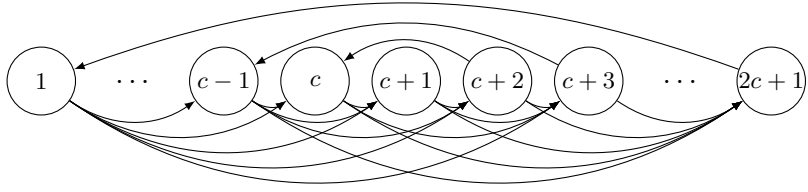
\begin{figure}[t]
\centering
\begin{tikzpicture}
\tikzstyle{vertex}=[draw,circle,minimum size=0.9cm,inner sep=0pt]
\node[vertex] (u) at (0,0) {{\footnotesize{$c+1$}}};
\node[vertex] (1) at (1.2,0) {{\footnotesize{$c+2$}}};
\node[vertex] (2) at (2.4,0) {{\footnotesize{$c+3$}}};
\node (i) at (3.6,0) {$\ldots$};
\node[vertex] (c) at (4.8,0) {{\footnotesize{$2c+1$}}};
\node[vertex] (-1) at (-1.2,0) {{\footnotesize{$c$}}};
\node[vertex] (-2) at (-2.4,0) {{\footnotesize{$c-1$}}};
\node (-i) at (-3.6,0) {$\ldots$};
\node[vertex] (-c) at (-4.8,0) {{\footnotesize{$1$}}};
\draw[->,>=latex] (1) to[out = 145,in = 35] (-1);
\draw[->,>=latex] (2) to[out = 152,in = 28] (-2);
\draw[->,>=latex] (c) to[out = 160,in = 20] (-c);

\foreach \n/\m in {
-c/-2,-c/-1,-c/u,-c/1,-c/2,
      -2/-1,-2/u,-2/1,     -2/c,
            -1/u,     -1/2,-1/c,
                  u/1, u/2, u/c,
                       1/2, 1/c,
                            2/c} {
\draw[->,>=latex] (\n) to [bend right=35] (\m);
}
\end{tikzpicture}
\caption{A $c$-cutwidth-minimal tournament \blue{$T$} with $2c+1$ vertices. \blue{In the layout $\pi$ depicted above, the set of backward arcs of $T$ forms the matching $\{(\onesmt{\pi}{c+1+i},\onesmt{\pi}{c+1-i})\colon i\in[1,c]\}$.}
}\label{fig:minTournament}
\end{figure}

\subsection{Upper bound for $c$-cutwidth-minimal semi-complete digraphs}

\blue{We now prove the quadratic bound on the sizes of $c$-cutwidth-minimal semi-complete digraphs.}

\begin{theorem}\label{thm:obstructions-semi}
For every positive integer $c$, every $c$-cutwidth-minimal semi-complete digraph has at most \redd{$24c^2+1$} vertices.
\end{theorem}
\begin{proof}
For $c = 0$ and $c = 1$ the theorem holds trivially,
since the only $0$-cutwidth-minimal semi-complete digraph is an isolated vertex and the only two $1$-cutwidth-minimal semi-complete digraphs are the following: two vertices connected by a pair of symmetric arcs, 
and an oriented triangle (directed cycle on $3$ vertices). 
Thus we may assume that $c > 1$. 
Let $D$ be a $c$-cutwidth-minimal semi-complete digraph. 
We may assume that $D$ is not a tournament, as otherwise Theorem~\ref{thm:obstructions-tour} applies proving the bound.

We say that a $c$-cutwidth-miminal semi-complete digraph $D$ is \emph{$c$-critical} if every semi-complete subdigraph $D'$ of $D$ resulting from the removal of an arc satisfies $\cw(D')<c$.
We claim that it is enough to prove the statement for \emph{$c$-critical} semi-complete digraphs. Indeed, observe that every $c$-critical semi-complete subdigraph $D'$ of $D$ satisfies 
$V(D)=V(D')$, because otherwise, by the minimality of $D$, we would have $c \leq \cw(D') \leq \cw(D[V(D')])<c$, a contradiction.
So assume that $D$ is a $c$-critical semi-complete digraph. We now claim that $\cw(D)=c$. 
To see this consider any semi-complete digraph $D'$ obtained by removing an arc~$e$. 
Since $D$ is not a tournament, such $D'$ exists and, by $c$-criticality, it satisfies $c \leq \cw(D) \leq \cw(D')+1 \leq c$. This implies $\cw(D)=c$.

Consider applying the algorithm of Proposition~\ref{prop:turing-kernel} to the semi-complete digraph $D$ with parameter $c-1$.
Observe that the first step of this algorithm is to compute a vertex ordering $\pi$ of $D$ of width at most $2c-2$ using the approximation algorithm of Theorem~\ref{thm:apx}.
In case the application of Theorem~\ref{thm:apx} does not return such an ordering, the algorithm terminates and concludes that $\cw(D)>c-1$.
Moreover, this conclusion may be drawn only in such case: if the algorithm of Theorem~\ref{thm:apx} succeeds in finding a vertex ordering $\pi$ of width at most $2c-2$, 
the algorithm of Proposition~\ref{prop:turing-kernel} returns a list of induced subdigraphs $D_1,\ldots,D_\ell$ satisfying the asserted properties. 
However, since $\cw(D)=c$ and $c\leq 2c-2$ due to $c\geq 2$, we may apply the algorithm of Proposition~\ref{prop:turing-kernel}
on any ordering $\pi$ of $D$ of width $c$, instead of the one obtained by applying algorithm of Proposition~\ref{prop:turing-kernel}.
Thus we guarantee that the algorithm always produces a list of induced subdigraphs $D_1,\ldots,D_\ell$ with the asserted properties:
for every $i\in[1,\ell]$, $D_i$ is an induced subdigraph of $D$ on at most \redd{$24(c-1)^2+40(c-1)+1 \leq 24c^2+1$} vertices, and $\cw(D)\leq c-1$ if and only if $\cw(D_i)\leq c-1$ for each $i\in [1,\ell]$.

If we now had that each output $D_i$ was smaller by at least one vertex than $D$, then by the $c$-cutwidth minimality of $D$ we would infer that $\cw(D_i)\leq c-1$ for each $i\in [1,\ell]$, implying
$\cw(D)\leq c-1$. This is a contradiction with the assumption $\cw(D) = c$.
Hence, for some $i\in [1,\ell]$ we have $D_i=D$, and hence \redd{$|V(D)| \leq 24c^2+1$}.
\end{proof}

Observe that in the previous proof, to apply the algorithm of Proposition~\ref{prop:turing-kernel}, we use the fact that $\cw(D)=c\leq 2c-2$ and then consider an ordering $\pi$ of optimal width $c$ instead of a $2$-approximation. Analyzing in details the algorithm of Proposition~\ref{prop:turing-kernel} in this context would lead to smaller constants in the bound on the size of $c$-cutwidth minimal semi-complete digraphs, namely $6c^2+O(c)$.

\subsection{Algorithmic applications} 

Consider the {\sc $c$-Cutwidth Vertex Deletion} problem defined as follows: 
given a digraph $D$ and integer $k$, decide whether it is possible to find a set $Z$ of at most $k$ vertices in $D$ such that $\cw(D-Z)\leq c$; here, $c$ is considered a fixed constant.
A set $Z$ with this property is called a {\em{deletion set to cutwidth at most $c$}}.

The upper bound on the sizes of $c$-cutwidth-minimal semi-complete digraphs, proved in the previous section, yields almost directly a number of algorithmic 
corollaries for the {\sc $c$-Cutwidth Vertex Deletion} problem.
More precisely, we show that for a fixed $c$, we can give
\begin{itemize}[nosep]
\item a single-exponential parameterized algorithm, with the running time tight under ETH;
\item an approximation algorithm with the approximation factor depending only on $c$; and
\item a polynomial kernelization algorithm.
\end{itemize}

As a preparation for these results, we first prove that it is possible to efficiently locate small obstacles for cutwidth at most $c$ in a semi-complete digraph of cutwidth larger than $c$.

\begin{lemma}\label{lem:find-minimal}
There exists an algorithm that given a semi-complete digraph $D$ on $n$ vertices and a nonnegative integer $c$, runs in time $2^{\Oh(\sqrt{c\log c})}\cdot n^3$ and either concludes that $\cw(D)\leq c$,
or finds an induced $(c+1)$-cutwidth-minimal subdigraph of $D$.
\end{lemma}
\begin{proof}
Recall that Fomin and Pilipczuk~\cite{FominP13,Pil13} gave an algorithm that verifies whether the cutwidth of a given $n$-vertex semi-complete digraph is at most $c$ in time $2^{\Oh(\sqrt{c\log c})}\cdot n^2$.
Given a semi-complete digraph $D$, we use this algorithm first to verify whether $\cw(D)\leq c$. 
If this is not the case, we perform the following procedure.

Set $D':=D$, iterate through the vertices of $D$, and for each consecutive vertex $u$ check using the algorithm of Fomin and Pilipczuk whether $\cw(D'-u)\leq c$.
If this is the case, then keep $u$ in $D'$ and proceed, and otherwise remove $u$ from $D'$ and proceed. 
Note that thus $D'$ changes over the course of the algorithm but stays an induced subdigraph of $D$.
We claim that $D''$, defined as $D'$ obtained at the end of the procedure, is $(c+1)$-cutwidth-minimal.
On one hand, we remove a vertex from $D'$ only when this does not lead to decreasing the cutwidth below $c+1$, so we maintain the invariant that the cutwidth of $D'$ is always larger than $c$.
On the other hand, each vertex $u$ we keep in $D'$ had the property that removing it would decrease the cutwidth to at most $c$ at the moment it was considered.
Since cutwidth is closed under induced subdigraphs, this is also true in $D''$, so indeed $D''$ is $(c+1)$-cutwidth-minimal.
The procedure applies the algorithm of Fomin and Pilipczuk  $n+1$ times, so the total running time is $2^{\Oh(\sqrt{c\log c})}\cdot n^3$.
\end{proof}

We remark that in the tournament setting, computing the cutwidth is a polynomial-time solvable problem by Theorem~\ref{thm:poly-time}.
By plugging this subroutine instead of the algorithm of Fomin and Pilipczuk, we infer that 
for tournaments, the algorithm of Lemma~\ref{lem:find-minimal} works in fully polynomial time, with no exponential multiplicative factor depending on $c$.

By applying a standard branching strategy, we obtain a single-exponential FPT algorithm for {\sc $c$-Cutwidth Vertex Deletion} in semi-complete digraphs.

\begin{theorem}\label{thm:cvd-fpt}
The {\sc $c$-Cutwidth Vertex Deletion} problem on semi-complete digraphs can be solved in time $2^{\Oh(\sqrt{c\log c})}\cdot c^{\Oh(k)}\cdot n^{\Oh(1)}$, 
where $n$ is the number of vertices of the input digraph and $k$ is the budget for the deletion set.
\end{theorem}
\begin{proof}
Let $D$ be the input digraph.
Run the algorithm of Lemma~\ref{lem:find-minimal} to either conclude that $\cw(D)\leq c$, or to find an induced $(c+1)$-cutwidth minimal semi-complete digraph $P$ in $D$.
By Theorem~\ref{thm:obstructions-semi}, $P$ has at most $\Oh(c^2)$ vertices.
Branch on which vertex of $P$ is included in the solution; that is, for each vertex $u$ of $P$ recurse on the digraph $D-u$ with budget $k-1$ for the deletion set.
Whenever the budget drops to $0$ and the considered digraph still does not have cutwidth at most $c$, we may discard the branch.
Conversely, if the budget is nonnegative and the considered digraph has cutwidth at most $k$, then we have found a deletion set to cutwidth at most $c$ of size at most $k$.
The recursion tree has depth at most $k$ and branching $\Oh(c^2)$, so the whole algorithm has running time $2^{\Oh(\sqrt{c\log c})}\cdot c^{\Oh(k)}\cdot n^{\Oh(1)}$.
\end{proof}

Similarly, we obtain a constant-factor approximation algorithm by greedily removing each encountered $(c+1)$-cutwidth-minimal induced subdigraph.

\begin{theorem}\label{thm:cvd-apx}
There exists an algorithm that given a semi-complete digraph $D$ on $n$ vertices and a nonnegative integer $c$, runs in time $2^{\Oh(\sqrt{c\log c})}\cdot n^{\Oh(1)}$ and 
finds a deletion set to cutwidth at most $c$ in $D$ whose size is at most $\Oh(c^2)$ times larger than the minimum size of a deletion set to cutwidth at most $c$ in $D$.
\end{theorem}
\begin{proof}
Let $D$ be the input digraph and let $Y$ be a minimum-size deletion set to cutwidth at most $c$ in $D$.
Iteratively run the algorithm of Lemma~\ref{lem:find-minimal} to either conclude that the current digraph, initially set to $D$, has cutwidth at most $c$, 
or to find an induced $(c+1)$-cutwidth minimal semi-complete digraph $P$ in it. In the first case break the iteration and return the currently accumulated solution $X$, 
while in the second case remove all the vertices of $P$ from the current digraph, include them in $X$ (initially set to be empty), and continue the iteration.
By the condition of breaking the iteration we have that the obtained set $X$ satisfies $\cw(D-X)\leq c$.
On the other hand, each of the removed $(c+1)$-cutwidth-minimal induced subdigraphs $P$ had to include at least one vertex from $Y$. 
Since each such $P$ has at most $\Oh(c^2)$ vertices by Theorem~\ref{thm:obstructions-semi},
we conclude that $|X|\leq \Oh(c^2)\cdot |Y|$.
\end{proof}

Next we show that Theorem~\ref{thm:obstructions-semi} combined with the Sunflower Lemma approach to kernelization yields a polynomial kernel for {\sc $c$-Cutwidth Vertex Deletion} of size $k^{\Oh(c^2)}$.
We will use the following variant of the Sunflower Lemma due to Fomin et al.~\cite{FominSV13}, which is tailored to applications in kernelization. 
Here, a set $Z$ is a {\em{hitting set}} of a set family $\mathcal{F}$ if every set from $\mathcal{F}$ has a nonempty intersection
with $Z$. Also, $Z$ is a {\em{minimal}} hitting set of $\mathcal{F}$ if every its proper subset is not a hitting set of $\mathcal{F}$.

\begin{lemma}[Lemma 3.2 of~\cite{FominSV13}]\label{lem:sunflower}
Let $d$ be a fixed integer. Let $\mathcal{F}$ be a family of subsets of some universe $U$, each of cardinality at most $d$.
Then, given an integer $k$, one can in time $\Oh(|\mathcal{F}|\cdot (k+|\mathcal{F}|))$ compute a subfamily $\mathcal{F}'\subseteq \mathcal{F}$ with $|\mathcal{F}'|\leq d!(k+1)^d$ such that the following holds:
every subset $Z\subseteq U$ of size at most $k$ is a minimal hitting set for $\mathcal{F}$ if and only if it is a minimal hitting set for $\mathcal{F}'$.
\end{lemma}

\begin{theorem}\label{thm:cvd-kernel}
Let $c$ be a fixed integer.
There exists an algorithm that given an instance $(D,k)$ of {\sc $c$-Cutwidth Vertex Deletion} where $D$ is a semi-complete digraph on $n$ vertices, runs in time $n^{\Oh(c^2)}$ and 
returns an equivalent instance $(D',k)$ of {\sc $c$-Cutwidth Vertex Deletion} where $D'$ is an induced subdigraph of $D$ on at most $k^{\Oh(c^2)}$ vertices.
\end{theorem}
\begin{proof}
Let $d=\Oh(c^2)$ be the upper bound on the number of vertices of any $(c+1)$-cutwidth-minimal semi-complete digraph, given by Theorem~\ref{thm:obstructions-semi}.
Find all $(c+1)$-cutwidth-minimal semi-complete digraphs in $D$ by iterating through all subsets $A$ of vertices of size at most $d$, and verifying whether $D[A]$ is $(c+1)$-cutwidth-minimal.
For a given set $A$, this can be done by applying the algorithm of Fomin and Pilipczuk~\cite{FominP13,Pil13} $|X|+1$ times: to check whether $\cw(D[A])>c$ and $\cw(D[A\setminus \{a\}])\leq c$ for each $a\in A$.
Since we inspect only induced subdigraphs of size at most $d$, this takes time $2^{\Oh(\sqrt{c\log c})}\cdot n^{d}$ in total.

Let $\mathcal{F}$ be the obtained family of vertex sets of $(c+1)$-cutwidth-minimal induced subdigraphs of $D$.
It is clear that for any $X\subseteq V(G)$, we have that $\cw(D[X])\leq c$ if and only if $X$ does not fully contain any set $A$ from $\mathcal{F}$.
Conversely, a set $Z$ is a deletion set to cutwidth at most $c$ in $D$ if and only if $Z$ is a hitting set for $\mathcal{F}$.

Apply the algorithm of Lemma~\ref{lem:sunflower} to the family $\mathcal{F}$, yielding a subfamily $\mathcal{F}'$ of size at most $d!\cdot k^{d+1}$ that has exactly the same minimal hitting sets of
size at most $k$. Let $W\subseteq V(G)$ be the union of all the sets in $\mathcal{F}'$; then $|W|\leq d\cdot |\mathcal{F}'|\leq d\cdot d!\cdot k^{d+1}$.
Denote $D'=D[W]$. We claim that the instance $(D,k)$ is a yes-instance of {\sc $c$-Cutwidth Vertex Deletion} if and only if $(D',k)$ does; note that then $D'$ can be output by the algorithm.
The left-to-right implication is trivial: intersecting with $W$ maps every solution to $(D,k)$ to a non-larger solution to $(D',k)$.

For the converse implication, suppose $Z\subseteq W$ is such that $|Z|\leq k$ and $\cw(D'-Z)\leq c$.
In particular, $Z$ is a hitting set of $\mathcal{F}'$, so let $\widehat{Z}\subseteq Z$ be any minimal hitting set of $\mathcal{F}'$ contained in $Z$.
Then $|\widehat{Z}|\leq k$ and by Lemma~\ref{lem:sunflower} we infer that $\widehat{Z}$ is a minimal hitting set for $\mathcal{F}$ as well.
Since hitting sets for $\mathcal{F}$ are exactly deletions sets to cutwidth at most $c$ in $D$, we infer that $\cw(D-\widehat{Z})\leq c$ and $(D,k)$ is a yes-instance.
\end{proof}

We remark that for tournaments, we may plug in the bound of Theorem~\ref{thm:obstructions-tour} instead of the bound of Theorem~\ref{thm:obstructions-semi}.
This yields a more explicit running time of
$(2c+2\lceil \sqrt{2c}\rceil+1)^{k}\cdot n^{\Oh(1)}$ for the FPT algorithm, and reduces the approximation factor to $\Oh(c)$ and the kernel size to $k^{\Oh(c)}$.

Finally, we prove that the asymptotic running time of the algorithm of Theorem~\ref{thm:cvd-fpt} is likely to be tight: under ETH, there is no algorithm with running time subexponential in $k$.
This follows from an adaptation of the standard reduction from the \textsc{Vertex Cover} problem parameterized by solution size to \textsc{Feedback Vertex Set}.

\begin{figure}[b]
\centering
\begin{tikzpicture}
\tikzstyle{vertex}=[draw,circle,minimum size=0.5cm,inner sep=0pt]

\begin{scope}[shift={(-4,0)}]
\node {$C_{1,0}$};
\foreach \n in {1,2,3}{
\node[draw, circle] (a\n) at (\n*120-120:0.8) {};
}
\foreach \n in {1,2,3}{
\pgfmathtruncatemacro{\label}{mod(\n,3)+1}
\draw[->,>=latex] (a\n) -- (a\label);
}
\end{scope}

\draw[dotted] (-2.25,-1.2) -- (-2.25,1.3);

\node {$C_{2,1}$};
\foreach \n in {1,2,3,4,5}{
\node[vertex] (\n) at (\n*72-72:1.2) {$v_\n$};
\node[vertex] (\n') at (\n+1.25,0) {$v_\n$};
}
\draw[->,>=latex] (1) -- (5);
\draw[->,>=latex] (1') to [bend right=35] (5');
\draw[->,>=latex] (5) -- (2);
\draw[->,>=latex] (5') to [bend right=35] (2');
\foreach \n in {1,2,3,4}{
\foreach \m in {1,2}{
\pgfmathtruncatemacro{\label}{mod(\n+\m-1,5)+1}
\draw[->,>=latex] (\n) -- (\label);
\draw[->,>=latex] (\n') to [bend right=35] (\label');
}
}
\node at (4.25,-1.2) {$C_{2,1}$~sorted by $\pi_i = v_i$};
\end{tikzpicture}
\caption{The circular tournaments $C_{1,0}$ and $C_{2,1}$. 
}\label{fig:CircularTournament}
\end{figure}
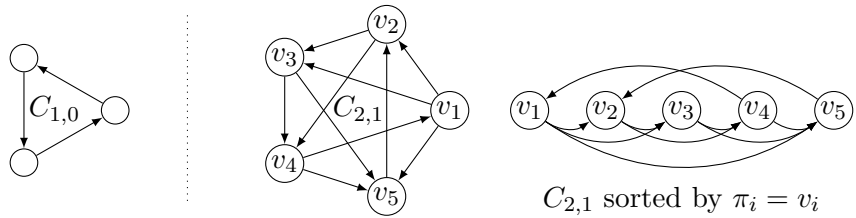

\begin{theorem}\label{thm:CVD-reduction}
For each fixed nonnegative integer $c$, the {\sc $c$-Cutwidth Vertex Deletion} problem is $\mathsf{NP}$-hard on tournaments. 
Moreover, unless ETH fails, it does not admit a $2^{o(n)}$- or a $2^{o(k)}\cdot n^{\Oh(1)}$-time algorithm in this setting.
\end{theorem}
\begin{proof}
Fix a nonnegative integer $c$.
We show a reduction from {\sc Vertex Cover} parameterized by solution size to the $c$-{\sc Cutwidth Vertex Deletion} problem with a linear blow-up of the instance size, 
and we ensure that the instance produced by the reduction is a tournament.
It is well-known (see e.g.~\cite{CFK14}) that {\sc Vertex Cover} has no $2^{o(n+m)}$ algorithm under ETH, where $n$ and $m$ denote the number of vertices and edges of the input graph, respectively.
Thus, the reduction will refute, under ETH, the existence of a $2^{o(n)}$-time algorithm for $c$-{\sc Cutwidth Vertex Deletion}, 
so in particular a $2^{o(k)}\cdot n^{\Oh(1)}$-time parameterized algorithm will be refuted.

For the reduction, we will need a bounded-size tournament with cutwidth exactly $c$. To show that such a tournament exists for each $c \in \mathbb{N}$, we propose the family of {\em{circular tournaments}} (see Figure \ref{fig:CircularTournament}) defined as follows. Given two nonnegative integers $t$ and $x$, the circular tournament $C_{t,x}$ has vertex set $V_t = \{v_1,v_2,...,v_{2t+1}\}$ and arc set $E_{t,x}$ such that for all $i,j \in [1,2t+1]$ with $i<j$, we have: \begin{itemize}
\item if $j \in [i+1,2t]$, then $(v_i,v_j) \in E_{t,x}$ whenever  $j-i \leq t$, otherwise $(v_j,v_i) \in E_{t,x}$;
\item if $j = 2t+1$, then $(v_i,v_j) \in E_{t,x}$ whenever  $i \in [1,x]$ or $j-i \leq t$, otherwise $(v_j,v_i) \in E_{t,x}$.
\end{itemize}

Let $t$ be the smallest integer such that $\frac {t(t+1)}2 \geq c$ and $x = \frac {t(t+1)}2 - c$; then $x<t$ by the minimality of $t$. 
It can be checked that then $(v_1,v_2,...,v_{2t+1})$ is a sorted vertex ordering of $C_{t,x}$, hence it is minimum by Lemma~\ref{lem:minimum-sorted}.
This implies that $\cw(C_{t,x})$ is equal to the width of this ordering, which is $\frac {t(t+1)}2 - x = c$.

Let $(G,k)$ be a parameterized instance of {\sc Vertex Cover}, we construct the parameterized instance $(T(G),k)$ of $c$-{\sc Cutwidth Vertex Deletion} as follows. Without loss of generality, we may assume that the vertices in $V$ are ordered, that is $V = \{v_1,v_2,...,v_n\}$. Then $T(G)$ is constructed as follows:\begin{itemize}[nosep]

\item For each $v_i \in V(G)$, add a copy $C^i_{t,x}$ of $C_{t,x}$ to $T(G)$ plus a fresh vertex $v'_i$; denote $V^i=V(C^i_{t,x})$. Moreover, for each vertex $u \in V^i$, add an arc $(v'_i,u)$ to $T(G)$, so that $v'_i$ has all the vertices of $C_{t,x}^i$ as outneighbors. 

\item For each $1 \leq i < j \leq n$, for every $u_i \in V^i$ and $u_j \in V^j$, add $(v'_i, u_j)$, $(u_i, u_j)$ and $(u_i, v'_j)$ to $T(G)$. Also, if $(v_i,v_j) \in E(G)$ then add $(v'_j,v'_i)$ to $T(G)$, and otherwise add $(v'_i,v'_j)$.
\end{itemize}

Clearly, $T(G)$ is a tournament with a linear number of vertices, which can be computed in polynomial time. We prove that $(G,k)$ is a yes-instance of {\sc Vertex Cover} if and only if $(T(G),k)$ is a yes-instance of $c$-{\sc Cutwidth Vertex Deletion}.

On one hand, let $W'$ be a solution to the instance $(T(G),k)$ of $c$-{\sc Cutwidth Vertex Deletion}. We construct the solution $W$ of $(G,k)$ for {\sc Vertex Cover} as follows: if there exists $v \in (\{v'_i\} \cup V^i) \cap W'$ then put $v_i$ into $W$. Clearly, $|W'|\leq |W|\leq k$. 
For the sake of a contradiction, suppose that $W$ is not a vertex cover of $G$. Let $v_iv_j$ be an edge of $G$ which is not covered by $W$.
Then, the tournament $T(G)-W'$ still contains the induced subtournament $T_{i,j} = T[V^i \cup \{v'_i, v'_j \}]$.
By construction, we have $d^-_{T_{i,j}}(v'_i) = 1$, $d^-_{T_{i,j}}(v'_j) = |V^i|$ and for every $u \in V^i$, we have $d^-_{T_{i,j}}(u) = 1+d^-_{C^i_{t,x}}(u) \leq |V^i|$. Therefore, there exists a sorted ordering $\pi$ of $T_{i,j}$ such that $\onesmt{\pi}{1} = v'_i$ and $\onesmt{\pi}{|V(T_{i,j})|} = v'_j$. Also, since for every $u \in V^i_t$ we have $d^-_{T_{i,j}}(u) = 1+d^-_{C^i_{t,x}}(u)$, the vertices of $V^i$ are sorted in $\pi$ as in a sorted ordering of $C^i_{t,x}$, so due to the edge $(v_j',v_i')$ we have that $\cw(T_{i,j},\pi) = 1 + \cw(C^i_{t,x})$.
By Lemma \ref{lem:minimum-sorted}, we infer a contradiction: $$\cw(T(G)-W') \geq \cw(T_{i,j}) = \cw(T_{i,j},\pi) = 1 + \cw(C^i_{t,x}) = 1+c.$$

On the other hand, let $W$ be a solution of $(G,k)$ for the {\sc Vertex Cover} problem. We naturally infer a solution $W'$ to the instance $(T(G),k)$ of $c$-{\sc Cutwidth Vertex Deletion} as follows: whenever $v_i \in W$, we put $v'_i$ into $W'$. 
By construction, for each $1 \leq i < j \leq n$, if there exists an arc from $u_j \in \{v'_j\} \cup V^j$ to $u_i \in \{v'_i\} \cup V^i$ in $T(G)$, then $u_j = v'_j$, $u_i = v'_i$ and $v_iv_j \in E(G)$. We deduce that either $v_i \in W$ and $v'_i \in W'$, or $v_j \in W$ and $v'_j \in W'$. In both cases, the arc $(v'_j,v'_i)$ is removed when deleting $W'$ from $T(G)$. 
Therefore, the remaining digraph $T(G)-W'$ contains no arc from $(\{v'_j\} \cup V^j)\setminus W'$ to $(\{v'_i\} \cup V^i)\setminus W'$, and it is easy to see that the digraph $T(G)-W'$ breaks into multiple strongly connected components, each either consisting of a single vertex $v'_i$, or being one of the gadgets $C^i_{t,x}$. Each of these strongly connected components has cutwidth at most $c$, so it follows that $\cw(T(G)-W')\leq c$.
\end{proof}

\section{Lower bounds}\label{sec:lower-bound}


In this section, we prove 
almost tight lower bounds for the complexity of computing the cutwidth and the \ola-cost of a semi-complete digraph. 
Precisely, we prove the following result.

\begin{theorem}\label{thm:lower-bound}
For semi-complete digraphs, both computing the cutwidth and computing the \ola{}-cost are $\mathsf{NP}$-hard problems.
Moreover, unless the Exponential Time Hypothesis fails:
\begin{itemize}[nosep]
\item the cutwidth cannot be computed in time $2^{o(n)}$ nor in time $2^{o(\sqrt{k})}\cdot n^{\Oh(1)}$; and
\item the \ola{}-cost cannot be computed in time $2^{o(n)}$ nor in time $2^{o(k^{1/3})}\cdot n^{\Oh(1)}$.
\end{itemize}
Here, $n$ is the number of vertices of the input semi-complete digraph, and $k$ is the target width/cost.
\end{theorem}

The proof of Theorem~\ref{thm:lower-bound} spans the whole remainder of this section.
We start our reduction from an instance of the {\sc{NAE-3SAT}} problem, which was defined in Section~\ref{sec:prelims} and for which a complexity lower bound under ETH is given by Corollary~\ref{cor:ETH}.

Let us introduce some notation. For a formula $\varphi$ in CNF, the variable and clause sets of $\varphi$ are denoted by $\vars(\varphi)$ and $\cls(\varphi)$, respectively.
A variable assignment $\alpha\colon \vars(\varphi)\to \{\bot,\top\}$ {\em{NAE-satisfies}} $\varphi$ if every clause of $\varphi$ has at least one, but not all literals satisfied.
Formula $\varphi$ is {\em{NAE-satisfiable}} if there is a variable assignment $\alpha$ that NAE-satisfies it; equivalently, both $\alpha$ and its negation $\neg \alpha$ satisfy $\varphi$.
A digraph is called {\em{basic}} if it is simple and has no pair of symmetric arcs.
For an integer $m > 0$, let $\lambda_m$ be the following function $(14m+1)$-tuple:
$$\lambda_m(i)=
\left\{\begin{array}{ll}
2i,& \text{when } i \in [0,5m]\\
5m+i,& \text{when } i \in [5m+1,6m]\\
11m,& \text{when } i \in [6m+1,7m]\\
18m-i,& \text{when } i \in [7m+1,12m]\\
42m-3i,& \text{when } i \in [12m+1,14m]
\end{array}\right.
$$
The following lemma encapsulates the first, main step of our reduction.

\begin{lemma}\label{lem:Reduction}
There exists a polynomial-time algorithm that, given a 3CNF formula $\varphi$ with $m$ clauses, returns a basic digraph $D(\varphi)$ with $14m$ vertices and $24m$ edges such that:
\begin{enumerate}[nosep]
\item for every vertex ordering $\pi$, we have $\cutvector{D(\varphi)}{\pi} \preceq \lambda_m$;
\item if $\varphi$ is NAE-satisfiable, then there exists a vertex ordering $\pi$ with $\cutvector{D(\varphi)}{\pi} = \lambda_m$;
\item if there is a vertex ordering $\pi$ with $\max\{\cutvector{D(\varphi)}{\pi}\}\geq 11m$, then $\varphi$ is NAE-satisfiable.
\end{enumerate}
\end{lemma}
\begin{proof}
Without loss of generality, we may assume that each clause of $\varphi$ contains exactly $3$ literals, by repeating some literal if necessary. 
Then, we may also assume that every variable of $\vars(\varphi)$ appears at least twice, because a variable that appears only once can always be set in order 
that the clause in which it appears is NAE-satisfied, and thus such a variable and its associated clause may be safely removed. 
For every variable $x\in \vars(\varphi)$, let $p_x$ be the number of occurrences of $x$ in the clauses of $\varphi$; 
hence $3m=\underset{x\in \vars(\varphi)}\sum ~p_x$ and $p_x\geq 2$ for each $x\in \vars(\varphi)$. 
We finally assume the clauses and literals are ordered, so we may say that a literal $\ell_x$ is the $i_x$th occurrence of variable $x$ in the clauses of $\varphi$, with $i_x \in [1,p_x]$.

We now describe the construction of $D(\varphi)$; see Figure~\ref{fig:construction}. 
For every variable $x\in \vars(\varphi)$ construct a {\em{variable gadget}} $G_x$, which is a directed cycle of length $2p_x$ with vertices named as follows:
$\bot^x_1 \to \top^x_1 \to \bot^x_2 \to \top^x_2 \to \ldots \to \bot^x_{p_x} \to \top^x_{p_x}\to \bot^x_1$.
Note that this cycle has no symmetric arcs since $p_x > 1$.

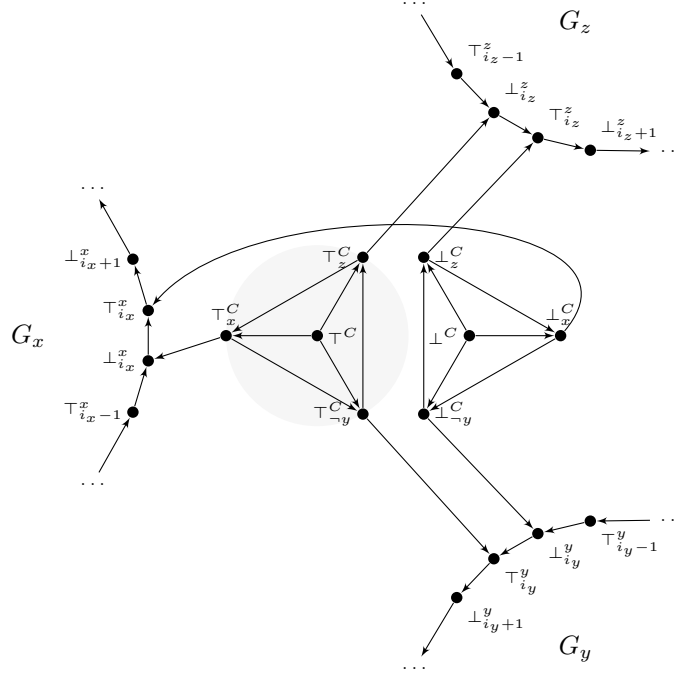
\begin{figure}[t]
\centering
\begin{tikzpicture}[scale=0.8]
   \tikzstyle{vertex}=[circle,fill=black,minimum size=0.15cm,inner sep=0pt]
   \fill[gray!7] (-1.25,0) circle (1.5);
   \foreach \n/\r/\p in {x/3/left,y/-1/below right,z/1/above right} {
     \begin{scope}[rotate={\r*60}, shift={(7,0)}]
      \node (c\n) at (-1,0) {{\small{$G_\n$}}};
      \node (u0\n) at (-50:-3.2) {{\tiny{$\ldots$}}};
      \node[vertex] (u1\n) at (-25:-3) {};
      \draw[\p] (u1\n) node {{\tiny{$\top^\n_{i_\n-1}$}}};
      \node[vertex] (u2\n) at (-8:-3) {};
      \draw[\p] (u2\n) node {{\tiny{$\bot^\n_{i_\n}$}}};
      \node[vertex] (u3\n) at (8:-3) {};
      \draw[\p] (u3\n) node {{\tiny{$\top^\n_{i_\n}$}}};
      \node[vertex] (u4\n) at (25:-3) {};
      \draw[\p] (u4\n) node {{\tiny{$\bot^\n_{i_\n+1}$}}};
      \node (u5\n) at (50:-3.2) {{\tiny{$\ldots$}}};
      \foreach \x/\y in {0/1,1/2,2/3,3/4,4/5} { 
      	\draw[thin,->,arrows={-latex'}] (u\x\n) -- (u\y\n);
      }
     \end{scope}
  }
  \node[vertex] (vx) at (-2.75,0) {};
  \draw[above] (vx) node {{\tiny{$\top^C_x$}}};
  \node[vertex] (vy) at (-.5,-1.3) {};
  \draw[left] (vy) node {{\tiny{$\top^C_{\neg y}$}}};
  \node[vertex] (vz) at (-.5,1.3) {};
  \draw[left] (vz) node {{\tiny{$\top^C_z$}}};
  \node[vertex] (w) at (-1.25,0) {};
  \draw[right] (w) node {{\tiny{$\top^C$}}};
  \node[vertex] (ux) at (2.75,0) {};
  \draw[above] (ux) node {{\tiny{$\bot^C_x$}}};
  \node[vertex] (uy) at (0.5,-1.3) {};
  \draw[right] (uy) node {{\tiny{$\bot^C_{\neg y}$}}};
  \node[vertex] (uz) at (0.5,1.3) {};
  \draw[right] (uz) node {{\tiny{$\bot^C_z$}}};
  \node[vertex] (uw) at (1.25,0) {};
  \draw[left] (uw) node {{\tiny{$\bot^C$}}};
  \foreach \x/\y in {w/vx,w/vy,w/vz,vx/vy,vy/vz,vz/vx,vx/u2x,vy/u3y,vz/u2z,uw/ux,uw/uy,uw/uz,ux/uy,uy/uz,uz/ux,uy/u2y,uz/u3z} {
	\draw[thin,->,arrows={-latex'}] (\x) -- (\y);
  }
  \draw[thin,->,arrows={-latex'}] (ux) to[out=50,in=50] (u3x);
  
\end{tikzpicture}
\caption{Part of the construction of $D(\varphi)$. 
The $\top$-clause gadget $G^C_\top$ (subgraph in the shaded circle) and the $\bot$-clause gadget $G^C_\bot$ for $C=x\vee \neg y\vee z$, 
together with neighboring parts of variable gadgets $G_x,G_y,G_z$.}\label{fig:construction}
\end{figure}

Then, for every clause $C\in \cls(\varphi)$, where $C=\ell_x\vee \ell_y\vee \ell_z$ for literals of variables $x,y,z\in \vars(\varphi)$, respectively, 
construct the following {\em{$\top$-clause gadget}} $G^C_\top$. Introduce a vertex $\top^C$ and a set of vertices $V^C_\top = \{\top^C_{\ell_x},\top^C_{\ell_y},\top^C_{\ell_z}\}$ 
together with the following arcs: 
\begin{itemize}[nosep]
\item A directed $3$-cycle $(\top^C_{\ell_x},\top^C_{\ell_y}),(\top^C_{\ell_y},\top^C_{\ell_z}),(\top^C_{\ell_z},\top^C_{\ell_x})$.
\item A set $\{(\top^C,\top^C_{\ell_x}),(\top^C,\top^C_{\ell_y}),(\top^C,\top^C_{\ell_z})\}$ of arcs from $\top^C$ to the vertices of $V_T^C$.
\end{itemize}
Similarly, construct the \emph{$\bot$-clause gadget} $G^C_\bot$, which is isomorphic to $G^C_\top$, but with vertices named $\bot$. 
Gadgets $G^C_\top$ and $G^C_\bot$ will differ in how we connect them with the rest of the graph.

Intuitively, the variable assignment $\alpha$, intended to NAE-satisfy $\varphi$, is encoded by choosing, in each variable gadget $G_x$, which vertices are placed in the first half of $\pi$, and which are placed in the second.
We use the gadget $G^C_\top$ to verify that $\alpha$ satisfies $C$,
whereas the gadget $G^C_\bot$ verifies that $\neg \alpha$ also satisfies $C$. 
For this purpose, connect the clause gadgets to variable gadgets as follows. Suppose $\ell_x \in C$ is the $i_x$th occurrence of $x$. If $\ell_x=x$ then add two arcs $(\top^C_{\ell_x},\bot^x_{i_x})$ and $(\bot^C_{\ell_x},\top^x_{i_x})$, and if $\ell_x=\neg x$ then add two arcs $(\top^C_{\ell_x},\top^x_{i_x})$ and $(\bot^C_{\ell_x},\bot^x_{i_x})$. 

This concludes the construction of $D(\varphi)$. Clearly $D(\varphi)$ is basic, and a straightforward verification using the equality $3m=\underset{x\in \vars(\varphi)}\sum ~p_x$ shows that conditions $|V(D(\varphi))| = 14m$ and $|E(D(\varphi))| = 24m$ hold as well.
We are left with verifying the three condition from the lemma statement. The following claim about the maximum size of a cut in $D(\varphi)$ will be useful.

\begin{claim}\label{cl:cut-any}
For any vertex subset $A\subseteq V$, we have that $|E(A,V\setminus A)|\leq 11m$.
\end{claim}
\begin{proof}
Denote $F=E(A,V\setminus A)$.
First, consider any variable $x\in \vars(\varphi)$. 
Since $G_x$ is a directed cycle of length $2p_x$, it can easily be seen that $|F \cap E(G_x)|\leq p_x$ and the equality holds if and only if $A$ contains every second vertex of the cycle $G_x$. 
Second, consider any clause $C=\ell_x\vee \ell_y\vee \ell_z \in \cls(\varphi)$. 
Let $R^C_\top$ be the set of three arcs connecting $G^C$ with the variable gadgets $G_x$, $G_y$, and $G_z$.
Since $\top^C$ has no incoming arcs, we can assume without loss of generality that $\top^C\in A$, as putting $\top^C$ into $A$ can only increase $|F|=|E(A,V\setminus A)|$. 
We now distinguish cases depending on the cardinality of $A \cap V^C_\top=A\cap \{\top^C_{\ell_x},\top^C_{\ell_y},\top^C_{\ell_z}\}$. 
The following implications follow from a straightforward analysis of the situation in $G^C_\top$ and on incident arcs.
\begin{itemize}
\item If $|A \cap V^C_\top|=0$ then $|F \cap R^C_\top|=0$ and $|F \cap E(G^C_\top)|=3$.
\item If $|A \cap V^C_\top|=1$ then $|F \cap R^C_\top|\leq 1$ and $|F \cap E(G^C_\top)|=3$.
\item If $|A \cap V^C_\top|=2$ then $|F \cap R^C_\top|\leq 2$ and $|F \cap E(G^C_\top)|=2$.
\item If $|A \cap V^C_\top|=3$ then $|F \cap R^C_\top|\leq 3$ and $|F \cap E(G^C_\top)|=0$.
\end{itemize}
In all the cases, we conclude that $|F \cap (E(G^C_\top)\cup R^C_\top)|\leq 4$; note that the equality can hold only in the two middle ones. The same analysis applies to the $\bot$-clause gadgets, yielding $|F \cap (E(G^C_\bot)\cup R^C_\bot)|\leq 4$, where $R^C_\bot$ is defined analogously. Since the set family 
$$\{E(G_x)\colon x\in \vars(\varphi)\}\cup \{(E(G^C_\top)\cup R^C_\top) \cup (E(G^C_\bot)\cup R^C_\bot) \colon C\in \cls(\varphi)\}$$
forms a partition of the edge set of $D(\varphi)$, we get:
\begin{align*}
|F| & =\sum_{x\in \vars(\varphi)} |F \cap E(G_x)|+\sum_{C\in \cls(\varphi)} |F \cap (E(G^C_\top)\cup R^C_\top \cup E(G^C_\bot)\cup R^C_\bot)|\\
& \leq \sum_{x\in \vars(\varphi)} p_x+\sum_{C\in \cls(\varphi)} 8=11m,
\end{align*}
which finishes the proof of the claim.
\cqed\end{proof}

We now verify the first condition from the lemma statement.

\begin{claim}\label{cl:pr1}
For any vertex ordering $\pi$ of $D(\varphi)$ and each position $i \in [0,14m]$, we have 
\begin{equation}\label{eq:main}
\cutvector{D(\varphi)}{\pi}(i) \leq \lambda_m(i).
\end{equation}
\end{claim}
\begin{proof}
Denote $V=V(D(\varphi))$ and $E=E(D(\varphi))$.
Observe that the in-degree of each vertex is at most two by construction. Thus, for any vertex ordering $\pi$ and each position $i$, we have
\begin{equation}\label{eq:1}
\cutvector{D(\varphi)}{\pi}(i) \leq |E(V,\fstsmt{\pi}{i})| \leq 2i.
\end{equation}
This verifies~\eqref{eq:main} for $i\in [0,5m]$.

Observe that, every independent set in $D(\varphi)$ has at most $5m$ vertices: at most $p_x$ vertices can be selected from each variable gadget $G_x$, for $x\in \vars(\varphi)$, and at most $1$ vertex from each clause gadget $G^C_\top$ and $G^C_\bot$, for $C \in \cls(\varphi)$. We deduce that every subset of $i$ vertices in $D(\varphi)$ induces a subdigraph with at least $i-5m$ arcs, and hence:
\begin{equation}\label{eq:2}
\cutvector{D(\varphi)}{\pi}(i) = |E(V,\fstsmt{\pi}{i})| - |E(\fstsmt{\pi}{i},\fstsmt{\pi}{i})| \leq 2i-(i-5m) = 5m+i.
\end{equation}
This verifies~\eqref{eq:main} for $i\in [5m+1,6m]$.

Reciprocally, the outdegree of each vertex of $D(\varphi)$ is at most three by construction, so
\begin{equation}\label{eq:3}
\cutvector{D(\varphi)}{\pi}(i) \leq |E(\lstsmt{\pi}{i},V)| \leq 3(|V|-i) = 42m-3i.
\end{equation}
This verifies~\eqref{eq:main} for $i\in [12m+1,14m]$.

To prove the next inequality, we make a finer analysis of vertices with outdegrees less than~$3$. 
For $1 \leq j \leq 3$ and a position $i$, let $V^j_i = \{u \in \lstsmt{\pi}{i} \colon d^+(u) = j\}$. 
Note that the set $V^3_i$ is a subset of the centers $\top^C$ and $\bot^C$ of the clause gadgets, thus $|V^3_i| \leq 2m$. 
We infer that
\begin{equation*}
|E(\lstsmt{\pi}{i},V)| = 3|V^3_i|+2|V^2_i|+|V^1_i| = 2|\lstsmt{\pi}{i}|+|V^3_i|-|V^1_i|\leq 2(|V|-i) + 2m - |V^1_i|.
\end{equation*}
Also, any independent set in $D(\varphi)[V^2_i \cup V^3_i]$ has size at most $2m$, since this digraph is an induced subdigraph of the union of the $2m$ clause gadgets. 
As before, we infer that this induced subdigraph has at least $|V^2_i|+|V^3_i|-2m$ edges, hence
\begin{equation*}
|E(\lstsmt{\pi}{i},\lstsmt{\pi}{i})| \geq |V^2_i|+|V^3_i|-2m = |V|-i-|V^1_i|-2m.
\end{equation*}
Piecing this altogether, we deduce 
\begin{equation}\label{eq:4}
\cutvector{D(\varphi)}{\pi}(i) = |E(\lstsmt{\pi}{i},V)| - |E(\lstsmt{\pi}{i},\lstsmt{\pi}{i})| \leq |V|+4m-i=18m-i.
\end{equation}
This verifies~\eqref{eq:main} for $i\in [7m+1,12m]$.

It remains to show the last inequality: for $i\in [6m+1,7m]$ it holds that
\begin{equation}\label{eq:5}
\cutvector{D(\varphi)}{\pi}(i) \leq 11m.
\end{equation}
This, however, follows immediately from applying Claim~\ref{cl:cut-any} to $A=\lstsmt{\pi}{i}$, so the proof of the claim is now complete.
\cqed\end{proof}

It remains to prove the last two conditions of the lemma. We proceed with the third one.

\begin{claim}\label{cl:pr3}
If there is a vertex ordering $\pi$ with $\max\{\cutvector{D(\varphi)}{\pi}\}\geq 11m$, then $\varphi$ is NAE-satisfiable.
\end{claim}
\begin{proof}
Suppose there exist a vertex ordering $\pi$ and position $i$ such that $\cutvector{D(\varphi)}{\pi}(i) \geq 11m$.
By applying Claim~\ref{cl:cut-any} to $A=\lstsmt{\pi}{i}$ we infer that also $\cutvector{D(\varphi)}{\pi}(i) \leq 11m$, so
$$|E^i_\pi|=|E(\lstsmt{\pi}{i},\fstsmt{\pi}{i})|=\cutvector{D(\varphi)}{\pi}(i)=11m$$
and all inequalities used in the proof of Claim~\ref{cl:cut-any} for $A=\lstsmt{\pi}{i}$ are in fact equalities.
In particular, from the examination of the proof it follows that:
\begin{itemize}
\item For every $x\in \vars(\varphi)$, $\lstsmt{\pi}{i} \cap V(G_x)$ is either $\{\bot^x_i\colon i \in [1,p_x]\}$ or $\{\top^x_i\colon i \in [1,p_x]\}$.
\item For every $C\in \cls(\varphi)$, either $|\lstsmt{\pi}{i} \cap V^C_\top|=1$ and $|E^i_\pi \cap R^C_\top|=1$, or $|\lstsmt{\pi}{i} \cap V^C_\top|=2$ and $|E^i_\pi \cap R^C_\top|=2$. In any case, $|E^i_\pi \cap R^C_\top|\geq 1$, and similarly $|E^i_\pi \cap R^C_\bot| \geq 1$.
\end{itemize}
Let $\alpha\colon \vars(\varphi)\to \{\bot,\top\}$ be an assignment defined as follows: for each $x\in \vars(\varphi)$, set $\alpha(x) = \top$ if $\lstsmt{\pi}{i} \cap V(G_x)=\{\top^x_i\colon i \in [1,p_x]\}$, and set $\alpha(x)=\bot$ otherwise. We prove that $\alpha$ NAE-satisfies~$\varphi$.

Take any $C\in \cls(\varphi)$. Since $|E^i_\pi \cap R^C_\top|\geq 1$, there must exist an arc $(\top^C_{\ell_x},\beta^x_i) \in R^C_\top$ such that $\top^C_{\ell_x} \in \lstsmt{\pi}{i}$ and $\beta^x_i \not \in \lstsmt{\pi}{i}$. By the definition of $\alpha$, the latter assertion is equivalent to $\alpha(x) = \neg \beta$. By the construction of $D(\varphi)$, the existence of an arc from $\top^C_{\ell_x}$ to $\beta^x_i$ exactly means that $\ell_x$ is satisfied by setting $\alpha(x)=\neg \beta$. Since $|E^i_\pi \cap R^C_\bot|\geq 1$ as well, a symmetric reasoning shows that there also exists a literal of $C$ which is not satisfied by $\alpha$. Therefore, both $\alpha$ and $\neg \alpha$ satisfy $\varphi$, which means that $\alpha$ NAE-satisfies $\varphi$.
\cqed\end{proof}

We are left with the second condition from the lemma statement.
\begin{claim}\label{cl:pr2}
If $\varphi$ is NAE-satisfiable, then there exists a vertex ordering $\pi$ with $\cutvector{D(\varphi)}{\pi} = \lambda_m$.
\end{claim}
\begin{proof}
 Suppose that $\varphi$ is NAE-satisfiable, and let $\alpha\colon \vars(\varphi)\to \{\bot,\top\}$ be a NAE-satisfying assignment for $\varphi$. We construct a vertex ordering $\pi$ such that $\cutvector{D(\varphi)}{\pi} = \lambda_m$ as follows. We list consecutive vertices in $\pi$ from left to right, thinking of introducing them along the ordering. Whenever we do not
specify the order of introducing some vertices, this order can be chosen arbitrarily.
\begin{itemize}
\item ($i \in [1,3m]$) For each $x \in \vars(\varphi)$, if $\alpha(x) = \top$, then introduce the vertices $\bot^x_i$ of $G_x$, or introduce the vertices $\top^x_i$ of $G_x$ otherwise.
\item ($i \in [3m+1,5m]$) For each clause $C \in \cls(\varphi)$, introduce a vertex $\top^C_{\ell_x} \in V^C_\top$ and a vertex $\bot^C_{\ell_y} \in V^C_\bot$ such that $\alpha$ satisfies $\ell_y$ but does not satisfy $\ell_x$.
Since $\alpha$ NAE-satisfies $\varphi$, such vertices exist.
\item ($i \in [5m+1,6m]$) Since we only introduced $2$ vertices per clause so far, for each clause $C \in \cls(\varphi)$ there remains an unused literal $\ell_z \in C$, i.e., a literal for which no corresponding vertex is already introduced. Introduce $\bot^C_{\ell_z}$ if $\alpha$ satisfies $\ell_z$, or $\top^C_{\ell_z}$ otherwise.
\item ($i \in [6m+1,7m]$) For each $C \in \cls(\varphi)$, introduce $\top^C_{\ell_z}$ if $\alpha$ satisfies $\ell_z$, or $\bot^C_{\ell_z}$ otherwise, where $\ell_z$ is as above.
\item ($i \in [7m+1,9m]$) For each $C \in \cls(\varphi)$, introduce $\bot^C_{\ell_x}$ and $\top^C_{\ell_y}$, where $\ell_x$ and $\ell_y$ are as three points above.
\item ($i \in [9m+1,12m]$) For each $x \in \vars(\varphi)$, if $\alpha(x) = \top$, introduce the vertices $\top^x_i$ of $G_x$, or introduce the vertices $\bot^x_i$ of $G_x$ otherwise.
\item ($i \in [12m+1,14m]$) For each $C \in \cls(\varphi)$, introduce $\top^C$ and $\bot^C$.
\end{itemize}
It is not hard to analyze the number of arcs that each introduced vertex $\onesmt{\pi}{i}$ brings and removes, when moving from $E^{i-1}_\pi$ to $E^{i}_\pi$.
In fact, in the interval $[1,5m]$ the cutsize is incremented by~$2$ with each new vertex, in $[5m+1,6m]$ it is incremented by $1$, in $[6m+1,7m]$ it stays the same,
in $[7m+1,12m]$ it decreases by $1$, and in $[12m+1,14m]$ it decreases by $3$.
This shows that the cut vector of $\pi$ is exactly equal to $\lambda_m$, as claimed.
\cqed\end{proof}

Claims~\ref{cl:pr1},~\ref{cl:pr3}, and~\ref{cl:pr2} together finish the proof of Lemma~\ref{lem:Reduction}.
\end{proof}

Note that Lemma~\ref{lem:Reduction} expresses a reduction from {\sc{NAE-3SAT}} to a maximization problem: NAE-satisfiability of $\varphi$ is equivalent to $D(\varphi)$ admitting a vertex ordering of width {\em{at least}} $11m$. 
\blue{By complementing the resulting digraph, we turn this maximization into a minimization problem.}
Precisely, given a simple digraph $D=(V,E)$, define its \emph{complement} as $\bar D = (V,\bar E)$, where $\bar E = V^2 \setminus (E\cup \{(u,u)\colon u\in V\})$.
That is, we take the complete digraph without self-loops on the vertex set $V$, and we remove all the arcs that are present in $D$.
Note that the complement of a basic digraph is semi-complete.

Now, let $\bar \lambda_m$ be the tuple such that for all $i \in [0,14m]$, we have $\lambda_m(i) + \bar \lambda_m(i) = i(14m-i)$. 
It is not hard to check that $\max\{\bar \lambda_m\} = \bar \lambda_m(7m) = 49m^2-11m$. A simple verification of how the conditions of
Lemma~\ref{lem:Reduction} are transformed under complementation yields the following.

\begin{lemma}\label{pro:Reduction}
The complement of $D(\varphi)$ is a semi-complete digraph $\bar D(\varphi)$ satisfying:
\begin{enumerate}[nosep]
\item for every vertex ordering $\pi$, we have $\bar \lambda_m \preceq \cutvector{\bar D(\varphi)}{\pi}$;
\item if $\varphi$ is NAE-satisfiable, then there exists a vertex ordering $\pi$ with $\cutvector{\bar D(\varphi)}{\pi} = \bar \lambda_m$;
\item if $\bar D(\varphi)$ admits a vertex ordering $\pi$ of width at most $49m^2-11m$, then $\varphi$ is NAE-satisfiable.
\end{enumerate}
\end{lemma}
\begin{proof}
Observe that for any digraph $D$ on $n$ vertices, any its vertex ordering $\pi$, and any $i \in [0,n]$, we have $\cutvector{D}{\pi}(i) + \cutvector{\bar D}{\pi}(i) = i(n-i)$, and then
\begin{equation}\label{eq:sum}
\cutvector{D(\varphi)}{\pi}(i) + \cutvector{\bar D(\varphi)}{\pi}(i) = \lambda_m(i) + \bar \lambda_m(i).
\end{equation}
Equality~\eqref{eq:sum} implies that $\bar \lambda_m(i) \leq \cutvector{\bar D(\varphi)}{\pi}(i)$ is equivalent to $\cutvector{D(\varphi)}{\pi}(i) \leq \lambda_m(i)$. By the first claim of Lemma~\ref{lem:Reduction}, we deduce that $\bar \lambda_m \preceq \cutvector{\bar D(\varphi)}{\pi}$ for every vertex ordering~$\pi$.
Similarly, if $\varphi$ is NAE-satisfiable, then by the second claim of Lemma~\ref{lem:Reduction} we have an ordering $\pi$ such that $\cutvector{D(\varphi)}{\pi}=\lambda_m(i)$,
and then $\cutvector{\bar D(\varphi)}{\pi}=\bar \lambda_m(i)$ due to~\eqref{eq:sum}.
Finally, if we have a vertex ordering $\pi$ of width at most $49m^2-11m$, that is, with $\max\{\cutvector{\bar D(\varphi)}{\pi}\} = 49m^2-11m$, then 
in particular $\cutvector{\bar D(\varphi)}{\pi}(7m) \leq 49m^2-11m$, which is equivalent to $\cutvector{D(\varphi)}{\pi}(7m) \geq 11m$, due to~\eqref{eq:sum}.
By the third claim of Lemma~\ref{lem:Reduction}, this implies that $\varphi$ is NAE-satisfiable.
\end{proof}

Thus, Lemma~\ref{pro:Reduction} shows that NAE-satisfiability of $\varphi$ is equivalent to $\bar D(\varphi)$ having cutwidth at most $49m^2-11m$.
However, the fact that NAE-satisfiability of $\varphi$ implies that  $\bar D(\varphi)$ admits a vertex ordering with a very concrete cut vector $\bar \lambda_m$, which is the best possible in the sense
of the first claim of Lemma~\ref{pro:Reduction}, also enables us to derive a lower bound for \ola. 
All these observations, together with the linear bound on the number of vertices of $\bar D(\varphi)$, make the proof of Theorem~\ref{thm:lower-bound} essentially complete. We now give a formal verification that Theorem~\ref{thm:lower-bound} follows from Lemmas~\ref{lem:Reduction} and~\ref{pro:Reduction}.

\begin{proof}[Proof of Theorem~\ref{thm:lower-bound}]
Suppose we are given an instance $\varphi$ of the {\sc{NAE-3SAT}} problem. Compute the semi-complete digraph $\bar D(\varphi)$; this takes polynomial time, and $\bar D(\varphi)$ has $14m$ vertices.
By Lemmas~\ref{lem:Reduction} and~\ref{pro:Reduction}, we have that $\varphi$ is NAE-satisfiable if and only if the cutwidth of $\bar D(\varphi)$ is at most $49m^2-11m$.
Indeed, if $\varphi$ is NAE-satisfiable, then $\bar D(\varphi)$ has an ordering with cut vector $\bar \lambda_m$, whose maximum is $49m^2-11m$, and the converse implication is exactly the third claim of Lemma~\ref{pro:Reduction}. Therefore, if we could verify whether a given semi-complete digraph $D$ on $n$ vertices has cutwidth at most $k$ in time $2^{o(n)}$ or $2^{o(\sqrt{k})}\cdot n^{\Oh(1)}$, then
by applying this algorithm to $\bar D(\varphi)$ we would resolve whether $\varphi$ is NAE-satisfiable in time $2^{o(m)}\cdot (n+m)^{\Oh(1)}$, contradicting ETH by Corollary~\ref{cor:ETH}.

For \ola, we adapt the argument slightly. Let $k=\sum\{\bar\lambda_m\}$; note that $k\in \Theta(m^3)$.
Again, by Lemmas~\ref{lem:Reduction} and~\ref{pro:Reduction} we infer that $\varphi$ is NAE-satisfiable if and only if the \ola-cost of $\bar D(\varphi)$ is at most $k$.
Indeed, if $\varphi$ is NAE-satisfiable, then $\bar D(\varphi)$ has an ordering with cut vector $\bar \lambda_m$, whose sum is $k$. 
On the other hand, by the first claim of Lemma~\ref{pro:Reduction} we have that every vertex ordering $\pi$ of $\bar D(\varphi)$ satisfies $\bar \lambda_m \preceq \cutvector{\bar D(\varphi)}{\pi}$,
so the only possibility of obtaining an ordering of cost at most $k$ is to have an ordering with cut vector equal to $\bar \lambda_m$.
This cut vector would in particular have maximum equal to $49m^2-11m$, which by the third claim of Lemma~\ref{pro:Reduction} implies that $\varphi$ is NAE-satisfiable.
As before, this means that using an algorithm for computing \ola-cost with running time $2^{o(n)}$ or $2^{o(\sqrt[3]{k})}\cdot n^{\Oh(1)}$ we would be able to resolve NAE-satisfiability of $\varphi$
in time $2^{o(m)}\cdot (n+m)^{\Oh(1)}$, contradicting ETH by Corollary~\ref{cor:ETH}.
\end{proof}

\section{Parameterization by the number of pure vertices}\label{sec:pure}

Consider the problems of computing the cutwidth and the \ola-cost of a semi-complete digraph.
\blue{On one hand,} the reduction of Lemma~\ref{lem:Reduction} constructs a basic digraph whose complement has a pair of symmetric arcs between almost every pair of vertices.
On the other hand, on tournaments the problems are polynomial-time solvable.
\blue{These observations suggest taking a closer look at the parameterization by the number of vertices incident to symmetric arcs, as this parameter is zero in the tournament case and high
in instances constructed in our hardness reduction.}
We indeed show that this parameterization leads to an FPT problem, even in a larger generality.
Call a vertex $u$ of a simple digraph $D$ {\em{pure}} if for any other vertex $v$, exactly one of the arcs $(u,v)$ or $(v,u)$ is present in $D$.
In this section we prove the following algorithmic result.

\begin{theorem}\label{thm:partitionFPT}
There is an algorithm that, given a simple digraph $D$ on $n$ vertices, computes the cutwidth and the \ola-cost of $D$ in time $2^k\cdot n^{\Oh(1)}$, where $k$ is the number of non-pure vertices in $D$.
The algorithm can also report orderings certifying the output values.
\end{theorem}

The proof of Theorem~\ref{thm:partitionFPT} is based on the fine understanding of minimum orderings that we developed in Section~\ref{sec:polynomial}.
Recall that in Lemma~\ref{lem:minimum-sorted} we have shown that in the case when all vertices are pure---that is, when the digraph is a tournament---minimum orderings are exactly the sorted ones.
We now extend this observarion to the case when not all vertices are pure, by showing that at least the pure ones may be sorted greedily according to their indegrees.

\begin{lemma}\label{lem:pure-ordering}
Let $P$ be the set of pure vertices of a simple digraph $D$, and let $\pi^P$ be a sorted ordering of $P$, that is, for each $u,v \in P$, if $d_D^-(u) < d_D^-(v)$ then $\pi^P(u) < \pi^P(v)$. There exists a $\cw$-optimal ordering $\pi$ of $D$ and an \ola-optimal ordering $\pi'$ of $D$ such that the restrictions of $\pi$ and $\pi'$ to $P$ are equal to $\pi^P$.
\end{lemma}
\begin{proof}
We only give the proof for cutwidth as the proof for \ola ~follows directly by changing each occurrence of $\cw$ by \ola. Let $\sigma$ be an arbitrary $\cw$-optimal vertex ordering of $D$. 
Suppose that the restriction $\sigma^P$ of $\sigma$ to $P$ differs from $\pi^P$. 
We show how to transform $\sigma$ into an $\cw$-optimal ordering $\pi$ of $D$ such that the restriction of $\pi$ to $P$ is $\pi^P$. 

Let $j$ be the smallest index such that $\onesmt{\pi^P}{j}\neq\onesmt{\sigma^P}{j}$. 
Let $u=\onesmt{\pi^P}{j}$ and $v=\onesmt{\sigma^P}{j}$.
By the minimality of $j$, the prefixes of $\pi^P$ and $\sigma^P$ up to position $j-1$ match, so $v$ must be later than $u$ in $\pi^P$, but $v$ is before $u$ in $\sigma^P$.
Since $\sigma^P$ is equal to $\sigma$ restricted to $P$, we have that $v$ is before $u$ in $\sigma$.
On the other hand, since $\pi^P$ is sorted with respect to indegrees in $D$, from $\pi^P(u)<\pi^P(v)$ we infer the inequality
\begin{equation}\label{eq:diamond}
d^-(u) \leq d^-(v).
\end{equation}

We define a vertex ordering $\sigma^*$ of $D$ as a copy of $\sigma$ except that we exchange the position of vertices $u$ and $v$.
Note that $\sigma^*$ and $\pi$ match on a prefix of length $i$, one longer than on which $\sigma$ and $\pi$ matched.
We claim that $\cw(D) = \cw(D,\sigma) = \cw(D,\sigma^*)$, that is, $\sigma^*$ is a $\cw$-optimal ordering of $D$. 
Note that if this was the case, then we could apply the same reasoning to $\sigma^*$ and further on, eventually obtaining an $\cw$-optimal ordering of $D$ whose restriction to $P$ is exactly equal to $\pi^P$.
This would conclude the proof.

It remains to prove the claim. Recall that $\sigma(v)<\sigma(u)$. 
By definition of a cut vector, for all $i\notin [\sigma(v),\sigma(u)]$, we have $\cutvector{D}{\sigma}(i) = \cutvector{D}{\sigma^*}(i)$. 
Let us consider the cuts at positions $i\in [\sigma(v),\sigma(u)]$. We have:
\begin{align}
\cutvector{D}{\sigma}(i) & = |E(\lstsmt{\sigma}{i},\fstsmt{\sigma}{i})|\nonumber \\ 
& = |E(V,\fstsmt{\sigma}{i})| - |E(\fstsmt{\sigma}{i},\fstsmt{\sigma}{i})|\nonumber \\
 & = |E(V,\fstsmt{\sigma}{i} \setminus \{v\})| + d^-(v) - |E(\fstsmt{\sigma}{i},\fstsmt{\sigma}{i})|\label{eq:sigma}
\end{align}
The same reasoning for $\sigma^*$ yields:
\begin{align}
\cutvector{D}{\sigma^*}(i) & = |E(V,\fstsmt{\sigma^*}{i}\setminus \{u\})| + d^-(u) - |E(\fstsmt{\sigma^*}{i},\fstsmt{\sigma^*}{i})|\label{eq:sigma-star}
\end{align}
Also, for any pure vertex $w$ belonging to any vertex set $W$, we have: \begin{align*}
|E(W,W)| & = |E(W \setminus \{w\},W \setminus \{w\})| + d_{D[W]}^+(w)+d_{D[W]}^-(w) - |E(x,x)| \\
& = |E(W \setminus \{w\},W \setminus \{w\})| + |W| - 1.
\end{align*}
Since $u$ and $v$ are pure vertices and $\fstsmt{\sigma}{i} \setminus \{v\} = \fstsmt{\sigma^*}{i} \setminus \{u\}$, it follows that:
\begin{align*}
|E(\fstsmt{\sigma}{i},\fstsmt{\sigma}{i})|& =|E(\fstsmt{\sigma}{i} \setminus \{v\},\fstsmt{\sigma}{i} \setminus \{v\})|+|\fstsmt{\sigma}{i} \setminus \{v\}|-1\\
& =|E(\fstsmt{\sigma^*}{i} \setminus \{u\},\fstsmt{\sigma^*}{i} \setminus \{u\})|+|\fstsmt{\sigma^*}{i} \setminus \{u\}|-1=|E(\fstsmt{\sigma^*}{i},\fstsmt{\sigma^*}{i})|
\end{align*}
By combining this with~\eqref{eq:sigma} and~\eqref{eq:sigma-star}, we conclude that 
\begin{align*}
\cutvector{D}{\sigma}(i) - \cutvector{D}{\sigma^*}(i) = d^-(v) - d^-(u),
\end{align*}
which is a positive value by~\eqref{eq:diamond}.

We conclude that $\cutvector{D}{\sigma^*}(i) \leq \cutvector{D}{\sigma}(i)$ for each $i \in [0,|V(D)|]$. 
Thus, we deduce that $\cw(D) \leq \cw(D,\sigma^*) \leq \cw(D,\sigma) = \cw(D)$, so $\sigma^*$ is indeed a $\cw$-optimal vertex ordering of $D$. This concludes the proof.
\end{proof}

We can now proceed to the proof of Theorem~\ref{thm:partitionFPT}.

\begin{proof}[Proof of Theorem~\ref{thm:partitionFPT}]
We focus on computing the cutwidth of $D$, at the end we will briefly argue how the algorithm can be adjusted for the \ola-cost.
Let $P$ and $Q$ be the sets of pure and non-pure vertices in $D$, respectively. 
Compute any ordering $\nu^P$ of $P$ that sorts the vertices of $P$ according to non-decreasing indegrees in $D$.
By Lemma~\ref{lem:pure-ordering}, we know that there exists a $\cw$-optimal vertex ordering $\pi$ of $D$ such that the restriction of $\pi$ to $P$ is equal to $\nu^P$.
We give a dynamic programming algorithm that attempts at reconstructing $\pi$ based on $\nu^P$.

The set space of the dynamic programming consists of pairs $(X,i)$, where $X$ is a subset of $Q$ and $i$ is an integer with $0\leq i\leq |P|$.
Thus, we have at most $2^k\cdot (n+1)$ states.
For a state $(X,i)$, let $S(X,i)=X\cup \nu^P[i]$ be the associated candidate for a prefix of $\pi$. We define the following value function for the states:
\begin{equation*}
\phi(X,i)=\min_{\sigma\colon \textrm{ordering of } S(X,i)}\ \left(\max_{j=1,2,\ldots,i}\{|E(V(D)\setminus \sigma[j],\sigma[j])|\}\right).
\end{equation*}
In other words, assuming that $S(X,i)$ is a prefix of the constructed ordering $\pi$, $\phi(X,i)$ tells us how small maximum cutsize we can obtain among cuts within this prefix.
It is straightforward to verify that function $\phi(X,i)$ satisfies the following recurrence, which corresponds to choosing whether the last vertex of the ordering of $S(X,i)$ belongs to $X$ or to $\nu^P[i]$.
\begin{eqnarray*}
\phi(\emptyset,0)=& 0 & \\
\phi(X,i)= & \max( &  |E(V(D)\setminus S(X,i),S(X,i))|,\\
& & \min(\{\phi(X,i-1)\}\cup \{\phi(X\setminus \{x\},i)\colon x\in X\})).
\end{eqnarray*}
Here, we use the convention that $\phi(X,-1)=+\infty$ for all $X\subseteq Q$.
Thus, we can compute the values of function $\phi(\cdot,\cdot)$ in the dynamic programming manner, by iterating on sets $X$ of increasing size and increasing $i$.
The computation of each value takes polynomial time and there are $2^k\cdot n$ values to compute, hence the running time follows.
The optimum value of cutwidth can be found as the value $\phi(Q,|P|)$, and an ordering certifying this value can be recovered in polynomial time using the standard method of back-links.
To see that the algorithm is correct, observe that since $\pi$ restricted to $P$ is equal to $\nu^P$, it follows that $\pi$ gives rise to a computation path of the dynamic programming above
that results in finding the optimum value of the cutwidth.

To adjust the algorithm to computing the OLA-cost of $D$, consider the adjusted value function for the states:
\begin{equation*}
\psi(X,i)=\min_{\sigma\colon \textrm{ordering of } S(X,i)}\ \left(\sum_{j=1}^i |E(V(D)\setminus \sigma[j],\sigma[j])|\right).
\end{equation*}
Then, $\psi(\cdot,\cdot)$ satisfies the following recurrence:
\begin{eqnarray*}
\psi(\emptyset,0)& = & 0 \\
\psi(X,i)& = & |E(V(D)\setminus S(X,i),S(X,i))|+\min(\{\psi(X,i-1)\}\cup \{\psi(X\setminus \{x\},i)\colon x\in X\}).
\end{eqnarray*}
Thus, the values of $\psi(\cdot,\cdot)$ can be computed in the same manner within the same time complexity. 
Again, the OLA-cost of $D$ is equal to $\psi(Q,|P|)$ and the ordering certifying this value can be recovered in polynomial time.
The argument for the correctness is the same.
\end{proof}

\section{Conclusions}\label{sec:conclusions}

In this work we have charted the computational complexity of cutwidth and \ola{} on semi-complete digraphs by proving almost tight algorithmic lower bounds 
under ETH and showing that cutwidth admits a quadratic Turing kernel, even though a classic polynomial kernel cannot be expected unless $\mathsf{NP}\subseteq \mathsf{coNP}/\textrm{poly}$.
A particular question that we leave open is whether the size of our Turing kernel for cutwidth could be improved to linear.
This might be suggested by the linear bound on the size of a $c$-cutwidth minimal tournament (Theorem~\ref{thm:obstructions-tour}).

Another interesting direction is to investigate further the complexity of graph modification problems related to cutwidth: 
apply at most $k$ modifications to the given semi-complete digraph in order to obtain a digraph of cutwidth at most $c$.
Some immediate corollaries for vertex deletions are discussed in  Section~\ref{sec:obstructions}, but it is also interesting to look at the arc reversal variant, where the allowed modification is reversing an arc.
For $c=0$, this problem is equivalent to the {\sc{Feedback Arc Set}} problem,
which has been studied intensively in tournaments and semi-complete digraphs~\cite{AlonLS09,BliznetsCKMP16,Faster,FominP13}. 
Further results on the vertex deletion and arc reversal problems related to cutwidth in semi-complete digraphs will be the topic
of a future paper, currently under preparation.


\bibliographystyle{abbrv}

\begin{thebibliography}{10}

\bibitem{AlonLS09}
N.~Alon, D.~Lokshtanov, and S.~Saurabh.
\newblock Fast {FAST}.
\newblock In {\em {ICALP} 2009}, volume 5555 of {\em Lecture Notes in Computer
  Science}, pages 49--58. Springer, 2009.

\bibitem{Bang-Jensen-Gutin}
J.~Bang{-}Jensen and G.~Gutin.
\newblock {\em Digraphs --- theory, algorithms and applications}.
\newblock Springer, 2002.

\bibitem{BFF12}
D.~Binkele-Raible, H.~Fernau, F.~V. Fomin, D.~Lokshtanov, S.~Saurabh, and
  Y.~Villanger.
\newblock Kernel(s) for problems with no kernel: {O}n out-trees with many
  leaves.
\newblock {\em ACM Transactions on Algorithms}, 8(4):38, 2012.

\bibitem{BliznetsCKMP16}
I.~Bliznets, M.~Cygan, P.~Komosa, L.~Mach, and M.~Pilipczuk.
\newblock Lower bounds for the parameterized complexity of {M}inimum {F}ill-in
  and other completion problems.
\newblock In {\em {SODA 2016}}, pages 1132--1151. {SIAM}, 2016.

\bibitem{CFS12}
M.~Chudnovsky, A.~Fradkin, and P.~D. Seymour.
\newblock Tournament immersion and cutwidth.
\newblock {\em Journal of Combinatorial Theory, Series B}, 102(1):93 -- 101,
  2012.

\bibitem{CS11}
M.~Chudnovsky and P.~Seymour.
\newblock A well-quasi-order for tournaments.
\newblock {\em Journal of Combinatorial Theory, Series B}, 101(1):47 -- 53,
  2011.

\bibitem{CFK14}
M.~Cygan, F.~V. Fomin, L.~Kowalik, D.~Lokshtanov, D.~Marx, M.~Pilipczuk,
  M.~Pilipczuk, and S.~Saurabh.
\newblock {\em Parameterized Algorithms}.
\newblock Springer, 2015.

\bibitem{DowneyF13}
R.~G. Downey and M.~R. Fellows.
\newblock {\em Fundamentals of Parameterized Complexity}.
\newblock Texts in Computer Science. Springer, 2013.

\bibitem{Dru15}
A.~Drucker.
\newblock New limits to classical and quantum instance compression.
\newblock {\em SIAM Journal of Computing}, 44(5):1443--1479, 2015.

\bibitem{Faster}
U.~Feige.
\newblock Faster {FAST} ({F}eedback {A}rc {S}et in {T}ournaments).
\newblock {\em CoRR}, abs/0911.5094, 2009.

\bibitem{FominP13}
F.~V. Fomin and M.~Pilipczuk.
\newblock Subexponential parameterized algorithm for computing the cutwidth of
  a semi-complete digraph.
\newblock In {\em {ESA 2013}}, volume 8125 of {\em Lecture Notes in Computer
  Science}, pages 505--516. Springer, 2013.

\bibitem{FominSV13}
F.~V. Fomin, S.~Saurabh, and Y.~Villanger.
\newblock A polynomial kernel for {P}roper {I}nterval {V}ertex {D}eletion.
\newblock {\em {SIAM} Journal on Discrete Mathematics}, 27(4):1964--1976, 2013.

\bibitem{Fra11}
A.~Fradkin.
\newblock {\em Forbidden structures and algorithms in graphs and digraphs}.
\newblock PhD thesis, Princeton University, 2011.

\bibitem{FradkinS15}
A.~Fradkin and P.~D. Seymour.
\newblock Edge-disjoint paths in digraphs with bounded independence number.
\newblock {\em Journal of Combinatorial Theory, Series B}, 110:19--46, 2015.

\bibitem{GW16}
V.~Garnero and M.~Weller.
\newblock Parameterized certificate dispersal and its variants.
\newblock {\em Theoretical Computer Science}, 622:66--78, 2016.

\bibitem{GPR16}
A.~C. Giannopoulou, M.~Pilipczuk, J.~Raymond, D.~M. Thilikos, and M.~Wrochna.
\newblock Cutwidth: Obstructions and algorithmic aspects.
\newblock In {\em {IPEC 2016}}, volume~63 of {\em LIPIcs}, pages 15:1--15:13.
  Schloss Dagstuhl---Leibniz-Zentrum f\"ur Informatik, 2016.

\bibitem{HKS13}
D.~Hermelin, S.~Kratsch, K.~So\l{}tys, M.~Wahlstr{\"{o}}m, and X.~Wu.
\newblock A completeness theory for polynomial ({T}uring) kernelization.
\newblock {\em Algorithmica}, 71(3):702--730, 2015.

\bibitem{IPZ01}
R.~Impagliazzo, R.~Paturi, and F.~Zane.
\newblock Which problems have strongly exponential complexity?
\newblock {\em Journal of Computer and System Sciences}, 63(4):512--530, 2001.

\bibitem{Jan17}
B.~M.~P. Jansen.
\newblock Turing kernelization for finding long paths and cycles in restricted
  graph classes.
\newblock {\em Journal of Computer and System Sciences}, 85:18--37, 2017.

\bibitem{Pilipczuk13}
M.~Pilipczuk.
\newblock Computing cutwidth and pathwidth of semi-complete digraphs via degree
  orderings.
\newblock In {\em STACS 2013}, volume~20 of {\em LIPIcs}, pages 197--208.
  Schloss Dagstuhl---Leibniz-Zentrum f\"{u}r Informatik, 2013.

\bibitem{Pil13}
M.~Pilipczuk.
\newblock {\em Tournaments and optimality: new results in parameterized
  complexity}.
\newblock PhD thesis, University of Bergen, Norway, 2013.

\bibitem{SKM12}
A.~Sch{\"a}fer, C.~Komusiewicz, H.~Moser, and R.~Niedermeier.
\newblock Parameterized computational complexity of finding small-diameter
  subgraphs.
\newblock {\em Optimization Letters}, 6(5):883--891, 2012.

\bibitem{Sch78}
T.~J. Shaefer.
\newblock The complexity of satisfiability problems.
\newblock In {\em STOC 1978}, pages 216--226. ACM, 1978.

\bibitem{Tho90}
R.~Thomas.
\newblock A {M}enger-like property of tree-width: the finite case.
\newblock {\em Journal of Combinatorial Theory Series B}, 48(1):67--76, 1990.

\bibitem{TTV17}
S.~Thomass{\'e}, N.~Trotignon, and K.~Vu\v{s}kovi\'{c}.
\newblock A polynomial {T}uring-kernel for {W}eighted {I}ndependent {S}et in
  bull-free graphs.
\newblock {\em Algorithmica}, 77(3):619--641, 2017.

\end{thebibliography}

\end{document}